\documentclass{CSML}

\pdfoutput=1

\usepackage{lastpage}

\def\dOi{13(3:23)2017}
\lmcsheading%
{\dOi}
{1--\pageref{LastPage}}
{}
{}
{Mar.~24, 2016}
{Sep.~13, 2017}
{}

\usepackage[hidelinks]{hyperref}
\usepackage{multirow,array}
\usepackage{amssymb,verbatim,color,amsthm}
\usepackage{slashbox}
\usepackage{subfigure}
\usepackage{amsmath}
\newcommand{\ZZ}{\mathbb{Z}}

\usepackage{fixmath}

\usepackage[utf8]{inputenc}
\usepackage[T1]{fontenc}

\usepackage[color=light gray]{todonotes}
\presetkeys{todonotes}{inline}{}

\makeatletter{}\newcommand{\Paragraph}[1]{\noindent\textbf{#1}}
\newcommand{\aut}{\mathcal{A}}

\newcommand{\PTIME}{\textsf{PTime}}
\newcommand{\PSPACE}{\textsf{PSpace}}
\newcommand{\EXPTIME}{\textsf{ExpTime}}
\newcommand{\NEXPTIME}{\textsf{NExpTime}}
\newcommand{\coNP}{\textsf{coNP}}

\newcommand{\undecidable}{undecidable}

\newcommand{\class}{\mathcal{C}}

\newcommand{\DFA}{\mathsf{DFA}}
\newcommand{\NFA}{\mathsf{NFA}}
\newcommand{\PDA}{\mathsf{PDA}}
\newcommand{\DPDA}{\mathsf{DPDA}}

\newcommand{\TED}{\mathsf{TED}}
\newcommand{\FED}{\mathsf{FED}}

\newcommand{\N}{\mathbb{N}}

\newcommand{\lang}{\mathcal{L}}

\newcommand{\R}{\textsf{R}}
\newcommand{\prefix}[1]{\textsf{prefix}(#1)}
\newcommand{\algoFEDPDADFA}{\textsf{InfEdsSeq}}

\definecolor{darkgreen}{RGB}{00,180,00}
\newcommand{\ed}{ed}
\newcommand{\wed}{wed}

\newtheorem{theorem}{Theorem}

\newtheorem{proposition}[theorem]{Proposition}

\newcommand{\NN}{\mathbb{N}}
\usepackage{tikz}

\usetikzlibrary{arrows,automata,shapes,decorations,calc,matrix,decorations.pathmorphing}

\title[Edit Distance for Pushdown Automata]{Edit Distance for Pushdown Automata}
\author[K.~Chatterjee]{Krishnendu Chatterjee\rsuper a}
\address{{\lsuper{a,b,c}}IST Austria}
\email{\{krishnendu.chatterjee, tah, rasmus.ibsen-jensen\}@ist.ac.at}

\author[T.~A.~Henzinger]{Thomas A. Henzinger\rsuper b}
\address{\vspace{-18 pt}}

\author[R.~Ibsen{-}Jensen]{Rasmus Ibsen{-}Jensen\rsuper c}
\address{\vspace{-18 pt}}

\author[J.~Otop]{Jan Otop\rsuper d}
\address{{\lsuper d}University of Wrocław}
\email{jotop@cs.uni.wroc.pl}

\thanks{This research was funded in part by the European Research Council (ERC) under grant 
agreement 267989 (QUAREM), by the Austrian Science Fund (FWF) projects S11402-N23 (RiSE) and Z211-N23 (Wittgenstein Award),
FWF Grant No P23499- N23, FWF NFN Grant No S11407-N23 (RiSE), 
ERC Start grant (279307: Graph Games), MSR faculty fellows award, and 
by the National Science Centre (NCN), Poland under grant 2014/15/D/ST6/04543.}

\begin{document}
\maketitle

\makeatletter{}\begin{abstract}
The edit distance between two words $w_1, w_2$ is the minimal number of word 
operations (letter insertions, deletions, and substitutions) 
necessary to transform $w_1$ to $w_2$.
The edit distance generalizes to languages $\lang_1, \lang_2$, where 
the edit distance from $\lang_1$ to $\lang_2$ is the minimal number $k$ such that for every word 
from $\lang_1$ there exists a word in $\lang_2$ with edit distance at most $k$.
We study the edit distance computation problem between pushdown automata 
and their subclasses.
The problem of computing edit distance to a pushdown automaton is undecidable,
and in practice, the interesting question is to compute the edit distance 
from a pushdown automaton (the implementation, a standard model for programs with 
recursion) to a regular language (the specification).
In this work, we present a complete picture of decidability and complexity
for the following problems: (1)~deciding whether, for a given threshold $k$, the edit distance from a pushdown automaton to a finite
automaton is at most $k$, and
(2)~deciding whether the edit distance from a pushdown automaton to a finite automaton is finite.
\end{abstract}

\section{Introduction}
\makeatletter{}\noindent{\em Edit distance.}
The edit distance~\cite{levenshtein1966binary} between two words is a well-studied 
metric, which is the minimum number of edit operations (insertion, deletion, 
or substitution of one letter by another) that transforms one word to 
another. 
The edit distance between a word $w$ to a language $\lang$ is the minimal 
edit distance between $w$ and words in $\lang$. 
The edit distance between two languages $\lang_1$ and $\lang_2$ is the supremum over 
all words $w$ in $\lang_1$ of the edit distance between $w$ and $\lang_2$.

\smallskip\noindent{\em Significance of edit distance.}
The notion of {\em edit distance} provides a quantitative measure of 
``how far apart'' are (a)~two words, (b)~words from a language, and (c)~two 
languages.
It forms the basis for quantitatively comparing sequences, a problem that 
arises in many different areas, such as error-correcting codes, 
natural language processing, and computational biology.
The notion of edit distance between languages forms the foundations of a 
quantitative approach to verification.
The traditional qualitative verification (model checking) question is the 
\emph{language inclusion} problem: given an implementation (source language) 
defined by an automaton $\aut_I$ and a specification (target language) 
defined by an automaton $\aut_S$, decide whether the language 
$\lang(\aut_I)$ is included in the language $\lang(\aut_S)$ 
(i.e., $\lang(\aut_I) \subseteq \lang(\aut_S)$). 
The \emph{threshold edit distance} ($\TED$) problem is a generalization of the 
language inclusion problem, which for a given integer threshold $k\geq 0$
asks whether every word in the source language $\lang(\aut_I)$ has edit 
distance at most $k$ to the target language $\lang(\aut_S)$ 
(with $k=0$ we have the traditional language inclusion problem).
For example, in simulation-based verification of an implementation against a 
specification, the measured trace may differ slightly from the 
specification due to inaccuracies in the implementation. 
Thus, a trace of the implementation may not be in the specification. 
However, instead of rejecting the implementation, one can quantify the 
distance between a measured trace and the specification.
Among all implementations that violate a specification, the closer the 
implementation traces are to the specification, the 
better~\cite{Chatterjee08quantitativelanguages,chatterjee2012nested,ModelMeasuring}.
The edit distance problem is also the basis for {\em repairing} 
specifications~\cite{riveros,boundedRiveros}. 

The $\TED$ problem answers a fine-grained question with a fixed bound on the number of edit operations. 
A related problem, the \emph{finite edit distance} ($\FED$) problem, asks
whether there exists $k\geq0$ such that the answer to the $\TED$ problem with threshold
$k$ is YES. Hence, in verification applications we ask the $\FED$ question first, 
and in case of the positive answer, we can ask the $\TED$ question.

\smallskip\noindent{\em Our models.}
In this work we consider the edit distance computation problem between two 
automata $\aut_1$ and $\aut_2$, where $\aut_1$ and $\aut_2$ can be (non-)deterministic finite 
automata or pushdown automata.
Pushdown automata are the standard models for programs with recursion, 
and regular languages are canonical to express the basic properties of systems 
that arise in verification.
We denote by DPDA (resp., PDA) deterministic (resp., non-deterministic) 
pushdown automata, and DFA (resp., NFA) deterministic 
(resp., non-deterministic) finite automata.
We consider source and target languages defined by DFA, NFA, DPDA, and PDA. 
We first present the known results and then our contributions.

\smallskip\noindent{\em Previous results.}
The main results for the classical language inclusion problem are as follows~\cite{HU79}: 
(i)~if the target language is a DFA, then it can be solved in polynomial time;
(ii)~if either the target language is a PDA or both source and target languages are DPDA, 
then it is undecidable; 
(iii)~if the target language is an NFA, then (a) if the source language
is a DFA or NFA, then it is $\PSPACE$-complete, and (b) if the source language 
is a DPDA or PDA, then it is $\PSPACE$-hard and can be solved in $\EXPTIME$ 
(to the best of our knowledge, there is a complexity gap where the upper 
bound is $\EXPTIME$ and the lower bound is $\PSPACE$). 
The $\TED$ and $\FED$ problems were studied for DFA and NFA. 
The $\TED$ problem is $\PSPACE$-complete, when the source and target languages are 
given by DFA or NFA~\cite{riveros,boundedRiveros}. 
When the source language is given by a DFA or NFA, the $\FED$ problem is: 
(i)~$\coNP$-complete, when the target language is given by a DFA~\cite{boundedRiveros},
(ii)~$\PSPACE$-complete, when the target language is given by an NFA~\cite{boundedRiveros}.

\begin{table}[t]
\begin{subtable}
\centering\centering
\begin{tabular}{|c|c|c|c|c|}
\hline
& $\class_2 = \DFA$ & $\class_2 = \NFA$  & $\class_2 = \DPDA$ & $\class_2 = \PDA$ \\ 
\hline
$\class_1 \in \{ \DFA, \NFA\}$ & \multirow{2}{*}{$\PTIME$} & \PSPACE-c   & 
 $\PTIME$ &  \\
\cline{1-1}
\cline{3-4}
{$\class_1 \in \{  \DPDA, \PDA\}$} &  & \textbf{\EXPTIME-c~(Th.~\ref{th:mainTED})} & 
\multicolumn{2}{c|}{ {\undecidable}} \\
\hline
\end{tabular}
\caption{Complexity of the language inclusion problem from $\class_1$ to $\class_2$. Our results are boldfaced.
}
\label{tab:complexityOfINC}
\end{subtable}

\begin{subtable}
\centering\centering
\begin{tabular}{|c|c|c|c|c|}
\hline
& $\class_2 = \DFA$ & $\class_2 = \NFA$  & $\class_2 = \DPDA$ & $\class_2 = \PDA$ \\ 
\hline
$\class_1 \in \{ \DFA, \NFA\}$ & \coNP-c~\cite{boundedRiveros} & \PSPACE-c~\cite{boundedRiveros}   & 
 open~(Conj.~\ref{conj:FEDisUndec}) &  \\
\cline{1-4}
\multirow{2}{*}{$\class_1 \in \{  \DPDA, \PDA\}$} &   \textbf{$\coNP$-complete}& \textbf{\EXPTIME-c} & 
\multicolumn{2}{c|}{\multirow{2}{*}{ \textbf{\undecidable~(Prop.~\ref{p:undecidable})}}} \\
&  \textbf{(Th.~\ref{th:FEDonDFAcoNP})}    & \textbf{(Th.~\ref{th:FEDmain})}  &   \multicolumn{2}{c|}{}\\
\hline
\end{tabular}
\caption{Complexity of $\FED(\class_1, \class_2)$. Our results are boldfaced.
}
\label{tab:complexityOfFED}
\end{subtable}

\begin{subtable}
\centering\centering
\begin{tabular}{|c|c|c|c|c|}
\hline
& $\class_2 = \DFA$ & $\class_2 = \NFA$  & $\class_2 = \DPDA$ & $\class_2 = \PDA$ \\ 
\hline
$\class_1 \in \{ \DFA, \NFA\}$ & \multicolumn{2}{c|}{ \PSPACE-c~\cite{riveros}}   & 
\textbf{\undecidable~(Prop.~\ref{th:fromDPDAUndecidable})} & \\
\cline{1-4}
{$\class_1 \in \{  \DPDA, \PDA\}$} &   \multicolumn{2}{c|}{\textbf{\EXPTIME-c (Th.~\ref{th:mainTED} (1))}} & \multicolumn{2}{c|}{\undecidable}   \\
\hline
\end{tabular}
\caption{Complexity of $\TED(\class_1, \class_2)$. Our results are boldfaced.\hspace{3cm}}
\label{tab:complexityOfTED}
\end{subtable}
\end{table}

\smallskip\noindent{\em Our contributions.}
Our main contributions are as follows.
\begin{enumerate}
\item We show that the $\TED$ problem is $\EXPTIME$-complete, when the source 
language is given by a DPDA or a PDA, and the target language is given by 
a DFA or NFA. 
We present a hardness result which shows that the $\TED$ problem is $\EXPTIME$-hard 
for source languages given as DPDA and target languages given as DFA. 
We present a matching upper bound by showing that for source languages given 
as PDA and target languages given as NFA the problem can be solved in $\EXPTIME$. 
As a consequence of our lower bound we obtain that the language inclusion 
problem for source languages given by DPDA (or PDA) and target languages given by NFA is $\EXPTIME$-complete.
In contrast, if the target language is given by a DPDA, then the $\TED$ 
problem is undecidable even for source languages given as DFA.
Thus we present a complete picture of the complexity of the $\TED$ problem, and 
in addition we close a complexity gap in the classical language inclusion problem.
Note that the interesting verification question is when the implementation 
(source language) is a DPDA (or PDA) and the specification (target language)
is given as a DFA (or NFA), for which we present decidability results
with optimal complexity.

\item We also study the $\FED$ problem.
For finite automata, it was shown in~\cite{riveros,boundedRiveros} that if the answer 
to the $\FED$ problem is YES, then a polynomial bound on $k$ exists.
In contrast, the edit distance can be exponential between DPDA and DFA.
We present a matching exponential upper bound on $k$ for the $\FED$ problem 
from PDA to NFA.
We show that when source languages are given as DPDA or PDA, 
the $\FED$ problem is:
(i)~$\coNP$-complete, if the target languages are given as DFA, and
(ii)~$\EXPTIME$-complete, if the target languages are given as NFA.
\end{enumerate}
The lower bound in (i) holds even for source languages given as DFA~\cite{boundedRiveros}.
Our results are summarized in Tables~\ref{tab:complexityOfINC}, \ref{tab:complexityOfFED}~and~\ref{tab:complexityOfTED}.

This paper extends~\cite{editDistanceConference} in the following two ways:
\begin{itemize}
\item We provide full proofs of all results from~\cite{editDistanceConference}.
\item We show that the $\FED$ problem is $\coNP$-complete if the source language is given by 
DPDA or PDA and the target language is an DFA. This result is technically involved, but it completes 
the complexity picture for the $\FED$ problem in case of 
the source language given by a pushdown automaton and the target language given by a finite automaton. 
\end{itemize}

\smallskip\noindent{\em Related work.}
Algorithms for edit distance have been studied extensively for words~\cite{levenshtein1966binary,AhoPeterson,Okuda,Pighizzini,Karp,Mohri}.
The edit distance between regular languages was studied in~\cite{riveros,boundedRiveros},
between timed automata in~\cite{timedEdit}, and between straight line programs 
in~\cite{lifshits2007processing,DBLP:conf/spire/Gawrychowski12}. 
A near-linear time algorithm to approximate the edit distance for a word to a 
{\sc Dyck} language has been presented in~\cite{Saha14}.

\section{Preliminaries}
\makeatletter{}\subsection{Words, languages and automata}
\newcommand{\dTree}{\mathcal{D}}

\Paragraph{Words.} 
Given a finite alphabet $\Sigma$ of letters, a \emph{word} $w$ is a finite sequence 
of letters.
For a word $w$, we define $w[i]$ as the $i$-th letter of $w$ and $|w|$ 
as its length.
For instance, if $w=abc$, then $w[2]=b$ and $|w|=3$.
We denote the set of all words over $\Sigma$ by $\Sigma^*$.
We use $\epsilon$ to denote the empty word.

\Paragraph{Pushdown automata.} 
A \emph{(non-deterministic) pushdown automaton} (PDA) is a tuple 
$(\Sigma, \Gamma, Q, S, \delta, F)$, where 
$\Sigma$ is the input alphabet, 
$\Gamma$ is a finite stack alphabet, $Q$ is a finite set of states, 
$S \subseteq Q$ is a set of initial states, 
$\delta \subseteq Q \times \Sigma \times (\Gamma \cup \{ \bot \}) \times Q \times \Gamma^*$
is a finite transition relation and $F \subseteq Q$ is a set of final 
(accepting) states. 
A PDA $(\Sigma, \Gamma, Q, S, \delta, F)$ is a \emph{deterministic pushdown automaton} (DPDA)
if $|S| = 1$ and $\delta$ is a function from $Q \times \Sigma \times (\Gamma \cup \{ \bot \})$ to $Q \times \Gamma^*$.
We denote the class of all PDA (resp., DPDA) by $\PDA$ (resp., $\DPDA$).
We define the size of a PDA $\aut = (\Sigma, \Gamma, Q, S, \delta, F)$, denoted by $|\aut|$,
 as $|Q| + |\delta|$.

\Paragraph{Runs of pushdown automata.} Given a PDA $\aut$ and a word $w = w[1] \ldots w[k]$ over $\Sigma$,
a \emph{run} $\pi$ of $\aut$ on $w$ is a sequence of elements from 
$Q\times \Gamma^*$ of length $k+1$ such that
$\pi[0] \in S \times \{ \epsilon \}$ and for every
$i \in \{1, \ldots, k\}$ either (1)~$\pi[i-1] = (q,\epsilon)$,
$\pi[i] = (q',u')$ and $(q, w[i], \bot, q', u') \in \delta$, or 
(2)~$\pi[i-1] = (q,ua)$,
$\pi[i] = (q',uu')$ and $(q, w[i], a, q', u') \in \delta$.
A run $\pi$ of length $k+1$ is \emph{accepting} if $\pi[k] \in F \times \{ \epsilon \}$, i.e.,
the automaton is in an accepting state and the stack is empty. The \emph{language recognized (or accepted) by $\aut$}, denoted $\lang(\aut)$, is the set of words that have an accepting run.

\Paragraph{Context free grammar (CFG).}  
A context free grammar (CFG) is a tuple $(\Sigma,V,S,P)$, where $\Sigma$ is the alphabet, $V$ is a set of {\em non-terminals}, $S \in V$ is a {\em start symbol} and $P$ is a set of {\em production rules}.
A production rule $p$ has the following form $p: A\rightarrow u$, where $A\in V$ and $u \in (\Sigma \cup V)^*$.

A CFG in Chomsky normal form (CNF) is the special case in which each production rule $p$ has one of the following forms (recall that $S$ is the start symbol): (1)~$p: A\rightarrow BC$, where $A\in V$ and $B,C\in V\setminus \{S\}$; or (2)~$p:A\rightarrow \alpha$, where $A\in V$ and $\alpha\in\Sigma$; or (3)~$p: S\rightarrow \epsilon$. It is well-known that any CFG can be brought onto CNF in polynomial time~\cite{C59}.

\Paragraph{Languages generated by CFGs.} Fix a CFG $G=(\Sigma,V,S,P)$. We define \emph{derivation} $\rightarrow_G$ as a relation on $(\Sigma \cup V)^*\times (\Sigma \cup V)^*$ as follows:
$w \rightarrow_G w'$ iff $w = w_1 A w_2$, with $A \in V$, and $w' = w_1 u w_2$ for some $u \in (\Sigma \cup V)^*$ such that $A \rightarrow u$ is a production from $G$.
We define $\rightarrow_G^*$ as the transitive closure of $\rightarrow_G$. The \emph{language generated by $G$}, denoted by $\lang(G) = \{ w \in \Sigma^* \mid  S \rightarrow_G^* w \}$ is the set of words that can be derived from $S$.  
We omit $G$ and write $\rightarrow^*$ for $\rightarrow_G^*$ if $G$ is clear from the context and for any non-terminal $A$ and word $w\in (\Sigma \cup V)^*$, we call $A\rightarrow^* w$ an {\em implied production rule}.
For instance, the CFG $G = (\Sigma,V,S,P)$, where $\Sigma = \{a,b\}$,  $V=\{S\}$, and the rules $P$ are $S\rightarrow a S b$ and $S\rightarrow ab$, generates the language $\{a^nb^n\mid n\geq 1\}$.

 It is well-known~\cite{HU79} that CFGs and PDAs are language-wise polynomial equivalent (i.e., there is a polynomial time procedure that, given a PDA,  outputs a CFG of the same language and vice versa).

\Paragraph{Derivation trees of CFGs.} 
Fix a CFG $G = (\Sigma,V,S,P)$. The CFG defines a (typically infinite) set of {\em derivation trees}. 
A derivation tree is an ordered tree\footnote{In an ordered tree, children 
of every node are ordered.} where (1)~each leaf is associated with an element of $\Sigma\cup V \cup \{\epsilon\}$; and (2)~each internal node $q$ is associated with a non-terminal $A \in V$ and production rule $p: A\rightarrow w$, such that $A$ has $|w|$ children and the $i$-th child, for each $i$, is associated with $w[i]$ if it is a leaf or
a production rule $p' : w[i] \rightarrow w'$ if it is an internal node.
A derivation tree $T$ defines a string $w(T)$ over $\Sigma\cup V$ formed by reading labels of the leaves of $T$ in an ascending lexicographic path order (``from left to right'') while skipping $\epsilon$ symbols. Existence of a derivation tree $T$ with the root $A$ certifies that 
$A \rightarrow_G^* w(T)$. 
For instance given $G$ (as in the previous paragraph), the  derivation tree for $aabb$ is as given in Figure~\ref{fig:ex}.

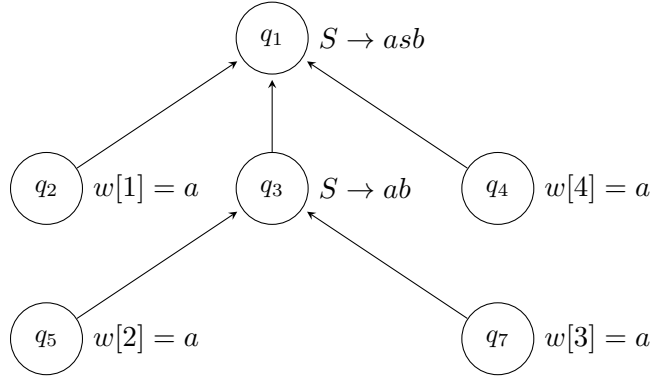
\begin{figure}
\center
\begin{tikzpicture}[node distance=2cm,-{stealth},shorten >=2pt]
\tikzstyle{every state}=[fill=white,draw=black,text=black,font=\small , inner sep=0.05cm]

\node[state,label=right:$S\rightarrow a s b$] (q1) {$q_1$};

\node[state,label=right:$S\rightarrow a b$,below of=q1] (q3) {$q_3$};
\node[state,label=right:{$w[1]=a$}] (q2) at ($(q3)+(-3cm,0)$) {$q_2$};
\node[state,label=right:{$w[4]=a$}] (q4)  at ($(q3)+(3cm,0)$) {$q_4$};

\node[state,label=right:{$w[2]=a$},below of=q2] (q5)  {$q_5$};
\node[state,label=right:{$w[3]=a$},below of=q4] (q7) {$q_7$};

\draw (q2) -> (q1);
\draw (q3) -> (q1);
\draw (q4) -> (q1);
\draw (q5) -> (q3);
\draw (q7) -> (q3);
\end{tikzpicture}
\caption{\label{fig:ex}Example of a derivation tree of $w=aabb$ for the CFG $G$ given in paragraph "Languages generated by CFGs."}
\end{figure}

\Paragraph{Finite automata.} A \emph{non-deterministic finite automaton} (NFA) is 
a pushdown automaton with empty stack alphabet. We will omit $\Gamma$ while referring to 
NFA, i.e., we will consider them as tuples $(\Sigma, Q, S, \delta, F)$.
We denote the class of all NFA by $\NFA$.
Analogously to DPDA we define \emph{deterministic finite automata} (DFA).

\Paragraph{Language inclusion.} Let $\class_1, \class_2$ be subclasses of $\PDA$. 
The \emph{inclusion problem from $\class_1$ in $\class_2$} asks, 
given $\aut_1 \in \class_1$, $\aut_2 \in \class_2$, whether $\lang(\aut_1) \subseteq \lang(\aut_2)$.

\Paragraph{Single letter operations on words.}
A single letter operation on a word can be either an {\em insertion}, a {\em deletion}, or a {\em substitution}.
Given a  letter $a\in \Sigma$ and a number $i$ we define relations $\rightarrow_{I(a,i)}, \rightarrow_{D(a,i)},
\rightarrow_{S(a,i)} \subseteq \Sigma^*\times \Sigma^*$ as follows

\begin{itemize}
\item 
the insert relation $\rightarrow_{I(a,i)}$: for all $w,w'$ we have $w\rightarrow_{I(a,i)}w'$ iff $w'=w[1]\dots w[i] a w[i+1]\dots w[|w|]$. For example, $abc\rightarrow_{I(a,2)}abac$. 
\item the delete relation $\rightarrow_{D(a,i)}$: for all $w,w'$ we have $w\rightarrow_{D(a,i)}w'$ 
iff $w'=w[1]\dots w[i-1] w[i+1]\dots w[|w|]$. For example, $abc\rightarrow_{D(b,2)}ac$. (Note that we ignore the letter parameter for deletions. We use $\rightarrow_{D(a,i)}$ over a notation like $\rightarrow_{D(i)}$ to ensure that all three types of single letter operations have 2 parameters)
\item the substitution relation $\rightarrow_{S(a,i)}$: for all $w,w'$ we have $w\rightarrow_{S(a,i)}w'$ 
iff $w'=w[1]\dots w[i-1]aw[i+1]\dots w[|w|]$. For example, $abc\rightarrow_{S(a,2)}aac$.
\end{itemize}

\Paragraph{Edit distance between words.} Given two words $w_1, w_2$, the edit 
distance between $w_1, w_2$, denoted by $\ed(w_1, w_2)$, is the minimal number of single letter operations:
insertions, deletions, and substitutions, necessary to transform $w_1$ into $w_2$.
More formally, $k:=\ed(w_1,w_2)$ is the length of the shortest sequence $S_1S_2\dots S_k$, 
 where each $S_j$ is an operation $S_j = (P_j,a_j,i_j)\in\{I,D,S\}\times \Sigma\times \NN$ for each $j$, such that there exist words $s_i$, $i\in \{0,\dots,k\}$, for which (1)~$w_1=s_0$, (2)~$w_2=s_k$ and (3)~$s_{j-1}\rightarrow_{P_j(a_j,i_j)}s_{j}$ for all $j\in \{1,\dots,k\}$.

\Paragraph{Edit distance between languages.} 
Let $\lang_1, \lang_2$ be languages. We define the edit distance 
\emph{from} $\lang_1$ \emph{to} $\lang_2$, denoted $\ed(\lang_1, \lang_2)$, as
$\sup_{w_1 \in \lang_1} \inf_{w_2 \in \lang_2} \ed(w_1, w_2)$.
The edit distance between languages is not a distance function. In particular,
it is not symmetric. 
For example: $\ed(\{a\}^*, \{a,b\}^*) = 0$, while
$\ed(\{a,b\}^*, \{a\}^*) = \infty$ because for every $n$, we have 
$\ed(\{ b^n \}, \{a\}^*) = n$.

\subsection{Problem statement}

In this section we define the problems of interest. Then, we recall the previous results 
and succinctly state our results. 

\begin{defi}
For  $\class_1, \class_2 \in \{ \DFA, \NFA, \DPDA, \PDA \}$ we define the following questions:
\begin{enumerate}
\item \emph{The threshold edit distance problem from $\class_1$ to $\class_2$ (denoted $\TED(\class_1,\class_2)$):} 
Given automata $\aut_1 \in \class_1$, $\aut_2 \in \class_2$ and an integer threshold $k\geq 0$, 
decide whether $\ed(\lang(\aut_1), \lang(\aut_2)) \leq k$.

\item
 \emph{The finite edit distance problem from $\class_1$ to $\class_2$ (denoted $\FED(\class_1,\class_2)$):}
Given automata  $\aut_1 \in \class_1$, $\aut_2 \in \class_2$,
decide whether $\ed(\lang(\aut_1), \lang(\aut_2)) < \infty$.

\item \emph{Computation of edit distance from $\class_1$ to $\class_2$:}
Given automata  $\aut_1 \in \class_1$, $\aut_2 \in \class_2$,
compute $\ed(\lang(\aut_1), \lang(\aut_2))$.
\end{enumerate}
\end{defi}

\noindent We establish the complete complexity picture for the $\TED$ problem for all combinations of 
source and target languages given by $\DFA, \NFA, \DPDA$ and $\PDA$:
\begin{enumerate}
\item $\TED$ for regular languages has been studied in~\cite{riveros}, where 
 $\PSPACE$-completeness of $\TED(\class_1, \class_2)$ for  $\class_1, \class_2 \in \{\DFA, \NFA\}$
has been established.
\item In Section~\ref{s:TEDPDAToRegular}, we study  the $\TED$ problem 
for source languages given by pushdown automata and target languages given 
by finite automata. 
We establish $\EXPTIME$-completeness of $\TED(\class_1, \class_2)$ for
 $\class_1 \in \{\DPDA, \PDA\}$ and $\class_2 \in \{\DFA, \NFA\}$.
\item In Section~\ref{s:fromPDA}, we study the $\TED$ problem for target languages
given by pushdown automata. We show that $\TED(\class_1, \class_2)$ is undecidable for 
$\class_1 \in \{\DFA, \NFA,\DPDA, \PDA\}$ and $\class_2 \in \{\DPDA, \PDA\}$.
\end{enumerate}

\noindent We study the $\FED$ problem for all combinations of 
source and target languages given by $\DFA, \NFA, \DPDA$ and $\PDA$ and obtain the following results:
\begin{enumerate}
\item $\FED$ for regular languages has been studied in~\cite{boundedRiveros}.
It has been shown that for $\class_1 \in \{ \DFA,\NFA\}$, the problem 
$\FED(\class_1,\DFA)$ is $\coNP$-complete, while the problem
$\FED(\class_1,\NFA)$ is $\PSPACE$-complete.
\item We show in Section~\ref{s:FEDPDAtoRegular} that for $\class_1 \in \{ \DPDA,\PDA\}$, 
the problem $\FED(\class_1,\NFA)$ is $\EXPTIME$-complete and
the problem $\FED(\class_1,\DFA)$ is $\coNP$-complete.
\item We show in Section~\ref{s:fromPDA} that
(1)~for $\class_1 \in \{\DFA, \NFA,\DPDA, \PDA\}$, the problem $\FED(\class_1, \PDA)$ is undecidable, and
(2)~the problem $\FED(\DPDA, \DPDA)$ is undecidable.
\end{enumerate}

\begin{rem}
{\bf Weighted edit-distance.} 
One could also consider a notion of weighted edit-distance, where a weight function $f:\{I,D,S\}\times \Sigma\rightarrow \ZZ$ is given that to each edit operation and letter assigns a weight.
I.e. inserting a letter $a$ might have a different weight from inserting a letter $b$.
The weighted edit-distance $\wed(w_1,w_2)$ would then be the minimum sum of weights $\sum_{j=1}^k f(P_j,a_j)$ over any $k$ and sequence of edit operation $S_1\dots S_k$,  
 where $S_j=(P_j,a_j,i_j)\in\{I,D,S\}\times \Sigma\times \NN$ for each $j$, such that there exists words $s_i$, $i\in \{0,\dots,k\}$, for which (1)~$w_1=s_0$, (2)~$w_2=s_k$ and (3)~$s_{j-1}\rightarrow_{P_j(a_j,i_j)}s_{j}$  for all $j\in \{1,\dots,k\}$.
 
 Our results extend to the case where $f$ assigns {\bf positive} weights. There are naturally no differences for the FED case (since if the minimum length is infinite, then so too is the sum of weights). There are no  differences either for the TED case, since the only time it comes up (in the following Claim~\ref{cla:letter-by-letter}) there are no differences.
 
 Allowing $f$ to assign zero or infinite weights leads to distances very different from the classical edit distance, such as the Humming distance,
 or the length difference. Such distances are out of scope of this paper.
\end{rem}

\section{Threshold edit distance from pushdown to regular languages}
\makeatletter{}\label{s:TEDPDAToRegular}

In this section we establish the complexity of the $\TED$ problem
from pushdown to finite automata. 

\begin{thm}
(1)~For $\class_1 \in \{\DPDA,\PDA\}$ and $\class_2 \in \{\DFA,\NFA\}$, the $\TED(\class_1, \class_2)$ problem  is $\EXPTIME$-complete.
(2)~For $\class_1 \in \{\DPDA,\PDA\}$, the language inclusion problem from $\class_1$ in $\NFA$ is $\EXPTIME$-complete.
\label{th:mainTED}
\end{thm}

We establish the above theorem as follows:
In Section~\ref{sec:upperBoundTED}, we present an exponential-time algorithm 
for $\TED(\PDA,\NFA)$ (for the upper bound of~(1)). 
Then, in Section~\ref{sec:lowerBoundTED} we show~(2), in a slightly stronger form, 
and reduce it (that stronger problem), to $\TED(\DPDA, \DFA)$, which shows the $\EXPTIME$-hardness 
part of (1).
We conclude this section with a brief discussion on parametrized complexity 
of $\TED$ in Section~\ref{sec:parametricTED}.

\subsection{Upper bound}
\label{sec:upperBoundTED}

We present an \EXPTIME\ algorithm that, given (1)~a PDA $\aut_P$; 
(2)~an NFA $\aut_N$; and (3)~a threshold $t$ given in  binary, decides whether 
the edit distance from $\aut_P$ to $\aut_N$ is above $t$.
The algorithm extends a construction for $\NFA$ by Benedikt et al.~\cite{riveros}.

\paragraph{Intuition.}
The construction uses the idea that for a given word $w$ and an NFA $\aut_N$ 
the following are equivalent: 
(i)~$\ed(w,\aut_N) > t$, and 
(ii)~for each accepting state $s$ of $\aut_N$ and for every word $w'$, if 
  $\aut_N$ can reach $s$ from some initial state upon reading $w'$, then $\ed(w,w') > t$.
We construct a PDA $\aut_I$ which simulates the PDA $\aut_P$ and stores in its states all
states of the NFA $\aut_N$ reachable with at most $t$ edits. 
More precisely, the PDA $\aut_I$ remembers in its states, for every state $s$ of the NFA $\aut_N$,
the minimal number of edit operations necessary to transform the currently read prefix $w_p$ of the input word 
into a word $w'_p$, upon which $\aut_N$ can reach $s$ from some initial state. 
If for some state the number of edit operations exceeds $t$, 
then we associate with this state a special symbol $\#$ to denote this.
Then, we show that a word $w$ accepted by the 
PDA $\aut_P$ has $\ed(w,\aut_N) > t$ iff the automaton $\aut_I$ has a run on $w$
that ends (1)~in an accepting state of simulated $\aut_P$, 
(2)~with the simulated stack of $\aut_P$ empty, and 
(3)~the symbol $\#$ is associated with every accepting state of $\aut_N$.

\begin{lem}
Given (1)~a PDA $\aut_P$; (2)~an NFA $\aut_N$; and (3)~a threshold $t$ given in binary, 
the decision problem of whether $\ed(\aut_P,\aut_N)\leq t$ can be reduced to the emptiness problem 
for a PDA of size $O(|\aut_P|\cdot (t+2)^{|\aut_N|})$.
\label{l:EDtoPDAreduction}
\end{lem}
\newcommand{\Impact}{\mathsf{Impact}}
\begin{proof}
\newcommand{\QN}{Q_{N}}
\newcommand{\FN}{F_{N}}
\newcommand{\QP}{Q_{P}}
\newcommand{\FP}{F_{P}}
\newcommand{\SP}{S_{P}}
\newcommand{\SN}{S_{N}}
\newcommand{\autImpact}{{\aut_{I}}}
\newcommand{\lpair}[1]{\langle #1 \rangle}
\newcommand{\deltaI}{\delta_I}
Let $\QN$ (resp., $\FN$) be the set of states (resp., accepting states) of $\aut_N$.
For $i \in \N$ and a word $w$, we define 
$T_w^i = \{ s \in \QN : $ there exists $w'$ with $\ed(w,w') = i$ such that $\aut_N $ has a run on the word $w'$ ending in $s \}$.
For a pair of states $s, s' \in \QN$ and $\alpha \in \Sigma\cup \{\epsilon\}$, we define  $m(s,s',\alpha)$ as the minimum number of edits needed to apply to $\alpha$ 
so that $\aut_N$ has a run on the resulting word from $s'$ to $s$.
For all $s, s' \in \QN$ and $\alpha \in \Sigma\cup \{\epsilon\}$, we
can compute $m(s,s',\alpha)$ in polynomial time in $|\aut_N|$.
For a state $s \in \QN$ and a word $w$ let $d_{w}^s=\min \{ i \geq 0 \mid s\in T_w^i \}$, i.e., $d_w^s$ is the minimal number of edits necessary 
to apply to $w$ such that $\aut_N$ reaches $s$  upon reading the resulting word.
We will  first prove the following claim.

\begin{clm}\label{cla:letter-by-letter}We have that $d_{w a}^s=\min_{s' \in \QN}(d_{w}^{s'}+m(s,s',a))$\end{clm}
\begin{proof}
Consider a run witnessing $d_{w a}^s$.
 As shown by~\cite{WF74} we can split the run into two parts, one sub-run on $w$ ending in $s'$, for some $s'$, and one sub-run on $a$ starting in $s'$. Clearly, the sub-run on $w$ has used $d_{w}^{s'}$ edits and the one on $a$ has used $m(s,s',a)$ edits.
\end{proof}

Let $\QP$ (resp., $\FP$) be the set of states (resp., accepting states) of the PDA $\aut_P$.
For every word $w$ and every state $q \in \QP$ such that there is a run on $w$ ending in $q$, we define 
$\Impact(w,q,\aut_P,\aut_N,t)$ as a pair 
$(q,\lambda)$ in $\QP \times \{0,1,\dots, t,\#\}^{|\QN|}$, where $\lambda$ is defined as follows:
for every $s \in \QN$ we have $\lambda(s) = d_w^{s}$ if $d_w^{s}\leq t$, and 
$\lambda(s) = \#$ otherwise. 
Clearly, the edit distance from $\aut_P$ to $\aut_N$ exceeds $t$ if there is a 
word $w$ and an accepting state $q$ of $\aut_P$ such that $\Impact(w,q,\aut_P,\aut_N,t)$ is a pair $(q,\lambda)$
and for every $s \in \FN$ we have $\lambda(s) = \#$
(i.e., the word $w$ is in $\lang(\aut_P)$ but any run of $\aut_N$ ending in $\FN$ has distance exceeding $t$).
 
We can now construct an {\em impact automaton}, a PDA $\autImpact$,
with state space $\QP \times \{0,1,\dots, t,\#\}^{\QN}$ and the
following transition relation:
A tuple $(\lpair{q,\lambda_1},a, \gamma, \lpair{q',\lambda_2}, u)$ is a transition of $\autImpact$ iff the following conditions hold:
\begin{enumerate}
\item the tuple projected to the first component of its state (i.e., the tuple $(q,a, \gamma, q', u)$) is a transition of $\aut_P$, and 
\item the second component $\lambda_2$ is computed from $\lambda_1$ using Claim~\ref{cla:letter-by-letter}, i.e.,
for every $s \in \QN$ we have $\lambda_2(s) = \min_{s' \in \QN} (\lambda_1(s')+m(s,s',a))$.
\end{enumerate}
The initial states of $\autImpact$ are $\SP \times \{ \lambda_0 \}$, where 
$\SP$ are initial states of $\aut_P$ and $\lambda_0$ is defined as follows.
For every $s \in \QN$ we have $\lambda_0(s) = \min_{s'\in \SN} m(s,s',\epsilon)$, where $\SN$ are initial states of $\aut_N$
(i.e.,  a start state of $\autImpact$ is a pair of a start state of $\aut_P$ together with the vector where the entry describing $s$ is the minimum number of edits needed to get to the state $s$ on the empty word). 
Also, the accepting states are $\{ \lpair{q,\lambda} \mid q\in \FP$ and for every $s \in \FN$ we have $\lambda(s) =\# \}$. 
Observe that for a run of $\autImpact$ on $w$ ending in $(s,\lambda)$, the vector $\Impact(w,s,\aut_P,\aut_N,t)$ is precisely $(s,\lambda)$.
Thus, the PDA $\autImpact$ accepts a word $w$ iff the edit distance between 
$\aut_P$ and $\aut_N$ is above $t$. 
Since the size of $\autImpact$ is $O(|\aut_P|\cdot (t+2)^{|\aut_N|})$ we obtain 
the desired  result.
\end{proof}

Lemma~\ref{l:EDtoPDAreduction} implies the following:

\begin{lem}\label{l:TEDinexp}
$\TED(\PDA, \NFA)$ is in $\EXPTIME$.
\end{lem}
\begin{proof}
Let $\aut_P, \aut_N$ and $t$ be an instance of $\TED(\PDA, \NFA)$, where
$\aut_P$ is a PDA, $\aut_N$ is an NFA, and $t$ is a threshold given in binary.
By Lemma~\ref{l:EDtoPDAreduction}, we can reduce $\TED$ to the emptiness question of 
a PDA of the size $O(|\aut_P|\cdot (t+2)^{|\aut_N|})$.
Since $|\aut_P|\cdot (t+2)^{|\aut_N|}$ is exponential in $|\aut_P| + |\aut_N| + t$ and
the emptiness problem for PDA can be decided in time polynomial in their size~\cite{HU79},
the result follows.
\end{proof}

\subsection{Lower bound}
\label{sec:lowerBoundTED}

Our $\EXPTIME$-hardness proof of $\TED(\DPDA,\DFA)$ extends the idea from~\cite{riveros} that shows $\PSPACE$-hardness of the edit distance for
DFA. The standard proof of $\PSPACE$-hardness of the universality problem
for $\NFA$~\cite{HU79} is by reduction to the halting problem of a fixed Turing machine $M$ working on a bounded tape. The Turing machine $M$ is the one that simulates other Turing machines (such a machine is called universal). 
The input to that problem is the initial configuration $C_1$ and the tape is bounded by its size $|C_1|$.  
In the reduction, the NFA recognizes the language of all words that do not encode a valid computation of $M$ starting from the initial configuration $C_1$, i.e., 
it accepts if one of the following conditions is violated: 
(1)~the given word is a sequence of configurations,
(2)~the state of the Turing machine and the adjacent letters follow from transitions of $M$, 
(3)~the first configuration is $C_1$
and (4)~the tape's cells are changed only by $M$, i.e., they do not change values spontaneously.
While violation of conditions (1), (2) and (3) can be checked by a DFA of polynomial size, condition~(4)
can be encoded by a polynomial-size NFA but not a polynomial-size DFA. However, to check~(4)
the automaton has to make only a single non-deterministic choice to pick a position in the encoding of the computation, 
which violates~(4), i.e., the value at that position is different from the value $|C_1|+1$ letters further, which
corresponds to the same memory cell in the successive configuration, and the head
of $M$ does not change it. We can transform a non-deterministic automaton $\aut_N$ checking (4) into 
a deterministic automaton $\aut_D$ by encoding such a non-deterministic pick using an external letter.
Since we need only one external symbol, we show that $\lang(\aut_N) = \Sigma^*$ iff
$\ed(\Sigma^*, \lang(\aut_D)) = 1$. This suggests the following definition:

\begin{defi}
An NFA $\aut = (\Sigma, Q, S, \delta, F)$ is \emph{nearly-deterministic} if 
$|S| = 1$ and $\delta = \delta_1 \cup \delta_2$, where $\delta_1$ is a function and in
every accepting run the automaton takes a transition from $\delta_2$ exactly once.
\end{defi}

\begin{lem}
There exists a DPDA $\aut_P$ such that the problem, given a nearly-deterministic NFA $\aut_N$, decide
whether $\lang({\aut_P}) \subseteq \lang(\aut_N)$, is $\EXPTIME$-hard.
\label{l:ExpTimeHardness}
\end{lem}

\begin{proof}
\newcommand{\Ap}{\mathcal{A}_P}
\newcommand{\Ar}{\mathcal{A}_N}
\newcommand{\UATM}{{M}_U}
Consider the \emph{linear-space halting} problem for a (fixed) alternating Turing machine (ATM) $M$:
given an input word $w$ over an alphabet $\Sigma$, decide whether $M$ halts on $w$ with the tape bounded by $|w|$.
There exists an ATM $\UATM$, such that
the linear-space halting problem for $\UATM$ is $\EXPTIME$-complete~\cite{Chandra:1981:ALT:322234.322243}.
We show the $\EXPTIME$-hardness of the problem from the lemma statement by reduction from the 
linear-space halting problem for $\UATM$.

Without loss of generality, we assume that existential and universal transitions of $\UATM$ alternate. Fix an input of length $n$.
The main idea is to construct the language $L$ of words that encode valid terminating computation trees of $\UATM$ on the given input. 
Observe that the language $L$ depends on the given input. 
We encode a single configuration of $\UATM$ as a word of length $n+1$ of the form $\Sigma^i q \Sigma^{n-i}$, where $q$ is a state of $\UATM$. 
Recall that a computation of an ATM is a tree, where every node of the tree is a configuration of $\UATM$,  
and it is accepting if every leaf node is an accepting configuration.
We encode computation trees $T$  of $\UATM$ by traversing $T$ in pre-order and executing the following:
if the current node has only one successor, then write down the current configuration $C$, terminate it with $\#$ 
and move down to the successor node in $T$.
Otherwise, if the current node has two successors $s,t$ in the tree, then write 
down in order (1)~the reversed current configuration $C^R$; and (2)~the results of traversals on $s$ and $t$, each surrounded  by parentheses $($ and $)$, i.e.,
$C^R\, (\, u^{s}\,)\,(\,u^{t}\,)\,$, where $u^{s}$ (resp., $u^{t}$) is the result of the traversal of the sub-tree of $T$ rooted at $s$ (resp., $t$).
Finally, if the current node is a leaf, write down the corresponding configuration and terminate with $\$$.
For example, consider a computation with the initial configuration $C_1$, from which an existential transition 
leads to $C_2$, which in turn has a universal transition to $C_3$ and $C_4$. Such a computation tree
 is encoded as follows:
\[
C_1\, \#\, C_2^R\, \left(\, C_{3} \ldots \$\,\right)\, \left(\, C_{4} \ldots \$\,\right). 
\]

We define automata $\Ar$ and $\Ap$ over the alphabet $\Sigma \cup \{\#,\$,(,)\}$. 
The automaton $\Ar$ is a nearly deterministic NFA that recognizes only (but not all) words 
not encoding valid computation trees of $\UATM$.
More precisely, $\Ar$ accepts in four cases:
(1)~The word does not encode a tree (except that the parentheses may not match as the automaton cannot check that) 
of computation as presented above.
(2)~The initial configuration is different from the one given as the input. 
(3)~The successive configurations, i.e., those that
result from existential transitions or left-branch universal transitions (like $C_2$ to $C_{3}$), are not valid. 
The right-branch universal transitions, which are preceded by the word ``$)($'', are not checked by $\Ar$. 
For example, the consistency of the transition $C_{2}$ to  $C_4$ is not checked by $\Ar$.
Finally, (4)~$\Ar$ accepts words in which at least one final configuration, which is a configuration followed by $\$$, is not final for $\UATM$.
Observe that conditions (1), (2) and (4) can be checked by polynomial-size DFA. Condition (3) can be checked by a polynomial-size nearly-deterministic NFA, which 
picks a position in $C_2$, for which the corresponding position in $C_3$ is faulty (either contains a spontaneous change of the corresponding tape cell or
 it is not compatible with any transition of $\UATM$). Picking such a position correspond to taking transition $\delta_2$ by a nearly-deterministic NFA.
 Thus, the automaton $\Ar$ is a nearly deterministic NFA, which recognizes the union of automata recognizing (1)-(4).
 
Next, we define $\Ap$ as a DPDA that accepts words in which parentheses match and right-branch universal transitions 
are consistent, e.g., it checks consistency of a transition from $C_2$ to $C_{4}$.
The automaton $\Ap$ pushes configurations on even levels of the computation tree (e.g., $C_2^R$), which are reversed, on the stack  and 
pops these configurations from the stack to compare them with the following configuration in the right sub-tree (e.g., $C_{4}$). 
In the example this means that, while the automaton processes 
the sub-word $\left(\, C_{3} \ldots \$\,\right)$, it can use its stack to check consistency of universal transitions in that sub-word.
We assumed that $\UATM$ does not have consecutive universal transitions. This means that, for example, $\aut_P$ does not need to check 
the consistency of $C_{4}$ with its successive configuration.
By construction, we have $L = \lang(\Ap) \cap \lang(\Ar)^c$ (recall that $L$ is the language of encodings of computations of $\UATM$ on the given input) and 
$\UATM$ halts on the given input if and only if $\lang(\Ap) \subseteq \lang(\Ar)$ fails.
Observe that $\Ap$ is fixed for all inputs, since it only depends on the fixed Turing machine $\UATM$.\end{proof}

Now, the following lemma, which is (2) of Theorem~\ref{th:mainTED}, 
follows from Lemma~\ref{l:ExpTimeHardness}.

\begin{lem}
The language inclusion problem from $\DPDA$ to $\NFA$ is $\EXPTIME$-complete.
\label{th:inlusionExpTimeHard}
\end{lem}
\begin{proof}
\newcommand{\Ap}{\mathcal{A}_p}
\newcommand{\Ar}{\mathcal{A}_r}
\newcommand{\Ad}{\mathcal{A}_d}
The $\EXPTIME$ upper bound is immediate (basically, an exponential 
determinization of the NFA, followed by complementation, 
product construction with the PDA, and the emptiness check of the product 
PDA in polynomial time in the size of the product).
$\EXPTIME$-hardness of the problem follows from Lemma~\ref{l:ExpTimeHardness}.
\end{proof}

Now, we show that the inclusion problem of DPDA in nearly-deterministic NFA,
which is $\EXPTIME$-complete by Lemma~\ref{l:ExpTimeHardness},
reduces to $\TED(\DPDA,\DFA)$. In the reduction, we transform a nearly-deterministic NFA $\aut_N$ over the alphabet $\Sigma$
into a DFA $\aut_D$ by encoding a single non-deterministic choice by auxiliary letters. 

\begin{lem}
$\TED(\DPDA,\DFA)$ is $\EXPTIME$-hard.
\label{th:TEDexpHard}
\end{lem}
\begin{proof}
To show $\EXPTIME$-hardness of $\TED(\DPDA,\DFA)$, 
we reduce the inclusion problem of $\DPDA$ in nearly-deterministic NFA to $\TED(\DPDA,\DFA)$.
Consider a DPDA $\aut_P$ and a nearly-deterministic NFA $\aut_N$ over an alphabet $\Sigma$.
Without loss of generality we assume that letters on even positions are $\diamondsuit \in \Sigma$ and $\diamondsuit$
do not appear on the odd positions.
Let $\delta = \delta_1 \cup \delta_2$ be the transition relation of $\aut_N$, where $\delta_1$ is a function
and along each accepting run, $\aut_N$ takes exactly one transition from $\delta_2$.
We transform the NFA $\aut_N$ to a DFA $\aut_D$ by extending the alphabet $\Sigma$ with external letters $\{ 1, \ldots, |\delta_2| \}$.
On letters from $\Sigma$, the automaton $\aut_D$ takes transitions from $\delta_1$. 
On a letter $i \in \{ 1, \ldots, |\delta_2| \}$, the automaton $\aut_D$ takes the $i$-th transition from $\delta_2$.

We claim that $\lang(\aut_P) \subseteq \lang(\aut_N)$ iff 
$\ed(\lang(\aut_P), \lang(\aut_D)) = 1$.
Every word $w \in \lang(\aut_D)$ contains a letter $i \in \{ 1, \ldots, |\delta_2| \}$, which does not belong to $\Sigma$. 
Therefore, $\ed(\lang(\aut_P), \lang(\aut_D)) \geq 1$.
But, if we substitute letter $i$ by the letter in the $i$-th transition of $\delta_2$, we get a word from $\lang(\aut_N)$. 
If we simply delete the letter $i$, we get a word which does not belong to $\lang(\aut_N)$ as it has letter $\diamondsuit$ on an odd position.
Therefore, $\ed(\lang(\aut_P), \lang(\aut_D)) \leq 1$ implies 
$\lang(\aut_P) \subseteq \lang(\aut_N)$. 
Finally, consider a word $w' \in \lang(\aut_N)$. The automaton $\aut_N$ has an accepting run on $w'$, which takes exactly once a transition from $\delta_2$.
Say the taken transition is the $i$-th transition and the position in $w'$ is $p$.
Then, the word $w$, obtained from $w'$ by substituting the letter at position $p$ by letter $i$, is accepted by $\aut_D$.
Therefore, $\lang(\aut_P) \subseteq \lang(\aut_N)$ implies $\ed(\lang(\aut_P), \lang(\aut_D)) \leq 1$.
Thus we have $\lang(\aut_P) \subseteq \lang(\aut_N)$ iff $\ed(\lang(\aut_P), \lang(\aut_D)) = 1$.
\end{proof}

\subsection{Parameterized complexity}
\label{sec:parametricTED}

Problems of high complexity can be practically viable if the complexity
is caused by a parameter, which tends to be small in the applications. In this section 
we discuss the dependence of the complexity of $\TED$ based on its input values.

\begin{proposition}
(1)~There exist a threshold $t > 0$ and a $\DPDA$ $\aut_P$ such that 
the variant of $\TED(\DPDA, \DFA)$, in which the threshold is fixed to $t$ and DPDA 
is fixed to $\aut_P$, is still $\EXPTIME$-complete.
(2)~The variant of $\TED(\PDA, \NFA)$, in which the threshold is given in unary and 
$\NFA$ is fixed, is in $\PTIME$.
\end{proposition}

\begin{proof}
\noindent\emph{(1):} The inclusion problem of DPDA in nearly-deterministic NFA is
$\EXPTIME$-complete even if a DPDA is fixed (Lemma~\ref{l:ExpTimeHardness}).
Therefore, the reduction in Lemma~\ref{th:TEDexpHard} works for threshold $1$ and fixed DPDA.

\noindent\emph{(2):} In the reduction from Lemma~\ref{l:EDtoPDAreduction}, the resulting PDA has size $|\aut_P|\cdot (t+2)^{|\aut_N|}$, where $\aut_P$ is a PDA, $\aut_N$ is an NFA
and $t$ is a threshold. If $\aut_N$ is fixed and $t$ is given in unary, then $|\aut_P|\cdot (t+2)^{|\aut_N|}$
is polynomial in the size of the input and we can decide its non-emptiness in polynomial time.
\end{proof}

Conjecture~\ref{conj1} completes the study of the parametrized complexity of $\TED$.

\begin{conj}\label{conj1}
The variant of $\TED(\PDA, \NFA)$, in which the threshold is given in binary and 
$\NFA$ is fixed, is in $\PTIME$.
\end{conj}

\section{Finite edit distance from pushdown to regular languages}
\makeatletter{}\label{s:FEDPDAtoRegular}
\newcommand{\nonTerm}{T}
\newcommand{\extNon}{B}
\newcommand{\reach}{\textsf{Reach}}

In this section we study the complexity of the $\FED$ problem 
from pushdown automata to finite automata.

\begin{thm}
(1)~For $\class_1 \in \{\DPDA,\PDA\}$ and $\class_2 \in \{\DFA,\NFA\}$ we have the following dichotomy:
for all $\aut_1 \in \class_1, \aut_2 \in \class_2$ either $\ed(\lang(\aut_1),\lang(\aut_2))$  is 
exponentially bounded in $|\aut_1| + |\aut_2|$ or $\ed(\lang(\aut_1),\lang(\aut_2))$ is infinite.
Conversely, for every $n$ there exist a DPDA $\aut_P$ and a DFA $\aut_D$, both of the size $O(n)$, such that 
$\ed(\lang(\aut_P),\lang(\aut_D))$ is finite and exponential in $n$ (i.e.,
the dichotomy is asymptotically tight).
(2)~For $\class_1 \in \{\DPDA,\PDA\}$ the $\FED(\class_1, \NFA)$ problem is $\EXPTIME$-complete.
(3)~For $\class_1 \in \{\DPDA,\PDA\}$ the $\FED(\class_1, \DFA)$ problem is $\coNP$-complete.
(4)~Given a PDA $\aut_P$ and an NFA $\aut_N$, we can compute the edit distance 
$\ed(\lang(\aut_P),\lang(\aut_N))$ in time exponential in $|\aut_P| + |\aut_N|$.
\label{th:FEDmain}
\end{thm}

First, we show in Section~\ref{sec:FEDUpperBoundNFA} the dichotomy of (1), 
which together with Theorem~\ref{th:mainTED}, implies the $\EXPTIME$ upper bound for (2).
Next, in Section~\ref{sec:FEDUpperBoundDFA}, we show that $\FED(\PDA, \DFA)$ problem is in $\coNP$, which together with the results from~\cite{boundedRiveros}
shows (3). 
Finally, in Section~\ref{sec:FEDLowerBound}, we show that $\FED(\DPDA, \NFA)$ is 
$\EXPTIME$-hard.
We also present the exponential lower bound for (1).
Conditions (1), (2), and Theorem~\ref{th:mainTED} imply (3)
(by iteratively testing with increasing thresholds upto exponential bounds 
along with the decision procedure from Theorem~\ref{th:mainTED}).

\subsection{Upper bound for NFA}
\label{sec:FEDUpperBoundNFA}

In this section we consider the problem of deciding whether the edit distance from a PDA to an NFA is finite. 

We first give an overview of the section. Let $\aut_N$ be an NFA and $\aut_P$ a PDA that has $\nonTerm$ non-terminals. We show (in Lemma~\ref{lem:comp_deco}) that for any word $w\in \lang({\aut_P})$ one can break the word into chunks $w=s_1 u_1 \ldots s_k u_k s_{k+1}$, such that $\sum_{i=1}^k |s_k|\leq 2^{\nonTerm}$ and for any $\ell$ word $w_\ell$ defined as $w_{\ell}=s_1 (u_1^\ell) \ldots s_k (u_k^\ell) s_{k+1}$ belongs to $\lang({\aut_P})$ (this is in some sense the opposite of the pumping lemma, since the part that {\em cannot} be pumped is small). We then show (this follows from Lemma~\ref{l:FED-equivalences}) that if there is a word $w \in \lang({\aut_P})$ such that $\ed(w,\lang({\aut_N}))>2^{\nonTerm}$, then for every word $w_\ell$ defined as above we have
 $\ed(w_{\ell+1},\lang({\aut_N}))>\ed(w_\ell,\lang({\aut_N}))$ for all $\ell\geq 0$, showing that the edit-distance $\ed(\lang({\aut_P}),\lang({\aut_N})$ is unbounded. On the other hand, clearly, if $\ed(w,\lang({\aut_N}))\leq  2^{\nonTerm}$ for all $w\in \lang({\aut_P})$, then the edit-distance $\ed(\lang({\aut_P}),\lang({\aut_N}))\leq  2^{\nonTerm}$ by definition.

We start with a reduction of the problem. 
Given a language $\lang$, we define 
$\prefix{\lang} = \{ u : u$ is a prefix of some word from $\lang \}$.
We call an automaton $\aut$ a \emph{safety automaton} if every state of $\aut$ is accepting. 
Note that automata are not necessarily total, i.e. some states might not have an outgoing transition for some input symbols, and thus a safety automaton does not necessarily accept all words.
Note that for every NFA $\aut_N$, the language $\prefix{\lang({\aut_N})}$ is the language of a safety NFA.
We show that $\FED(\PDA, \NFA)$ reduces to $\FED$ from $\PDA$ to safety NFA. 

\begin{lem}\label{lem:prefix_closure}
Let $\aut_P$ be a PDA and $\aut_N$ an NFA. The following inequalities 
hold: 
\[\ed(\lang({\aut_P}),\lang({\aut_N}))\geq \ed(\lang({\aut_P}),\prefix{\lang({\aut_N}))} \geq \ed(\lang({\aut_P}),\lang({\aut_N})) - |\aut_N|\]
\end{lem}
\begin{proof}
Since $\lang({\aut_N}) \subseteq \prefix{\lang({\aut_N})}$, we have 
\[\ed(\lang({\aut_P}),\lang({\aut_N})) \geq \ed(\lang({\aut_P}),\prefix{\lang({\aut_N})})\] as the latter is the minimum over a larger set by definition. 

Hence, we only need to show the other inequality.
First observe that for every $w \in \prefix{\lang({\aut_N})}$, upon reading $w$, the automaton $\aut_N$ can reach a state from which an accepting state is reachable and thus, an accepting state can be 
reached in at most $|\aut_N|$ steps. 
Therefore, for every $w \in \prefix{\lang({\aut_N})}$ there exists $w'$ of length bounded by $|\aut_N|$ such that 
$w w' \in \lang({\aut_N})$. It follows that $\ed(\lang({\aut_P}),\prefix{\lang({\aut_N})}) \geq \ed(\lang({\aut_P}),\lang({\aut_N})) - |\aut_N|$.
\end{proof}

\begin{rem}
Consider an NFA $\aut_N$ recognizing a language such that $\prefix{\lang(\aut_N)} = \Sigma^*$.
For every PDA $\aut_P$, the edit distance $\ed(\lang(\aut_P), \lang(\aut_N))$ is bounded by $|\aut_N|$.
\end{rem}

In the remainder of this section we work with context-free grammars (CFGs) instead of PDAs. There are polynomial-time transformations between CFGs and PDAs that preserve the generated language;
switching from PDAs to CFGs is made only to simplify the proofs. 
The following definition and lemma can be seen as a reverse version of the pumping lemma for context free grammars (in that we ensure that the part which can not be pumped is small). 

As an abuse of notation we will think of a sequence of words as both the concatenation of the words and the sequence.
We define $[1,k]=\{1,\dots,k\}$.

 \smallskip\noindent{\bf Left and right language. }
For a CFG $G$ and a non-terminal $A$, we define the languages 
\begin{align*}
\lang(G,A,L)&=\{w\in \Sigma^*\mid \exists w'\in \Sigma^*
              (A\rightarrow^* wAw') \} \text{\ \ and} \\
\lang(G,A,R)&=
\{w\in \Sigma^*\mid \exists w'\in \Sigma^* (A\rightarrow^* w'Aw) \}\enspace .
\end{align*}
Also, the set of directions $D$ is $D=\{L,R\}$.
We next argue that we can construct a CFG for $\lang(G,A,D)$.

\begin{lem}\label{lem:L(G,A,D)}
Given a CFG $G$, a non-terminal $A$ and a direction $D$, we can construct in polynomial time 
a CFG $G'$ for which $\lang(G')=\lang(G,A,D)$.
\end{lem}
\begin{proof}
We describe the construction of a CFG for $\lang(G,A,L)$ and the construction for $\lang(G,A,R)$ is similar.

To simplify, we consider $G$ to be on CNF. We construct $G'$ as follows: The CFG $G'$ consists of two versions of each non-terminal in $G$, one with a star and one without. I.e. for each non-terminal $X\in G$, we have the non-terminals $X$ and $X^*$ in $G^*$. The idea is that $X^*$ derives prefixes of words derivable from $X$ in $G$, which ends just before a $A$. 
The productions are then as follows: 
\begin{itemize}
\item {\bf Non-starred.} Each production of $G$ is also in $G'$, which defines the productions for the non-starred non-terminals.
\item {\bf Starred.} For each production $X\rightarrow BC$ in $G$, there are productions $X^*\rightarrow BC^*$ and $X^*\rightarrow B^*$ in $G'$.
\item {\bf Additional for $\mathbold{A^*}$.} The non-terminal $A^*$ has the production $A^*\rightarrow \epsilon$ in $G'$ (no other starred non-terminal can produce any terminal). 
\end{itemize}
The start symbol of $G'$ is $A^*$. 
We next argue that $\lang(G')=\lang(G,A,L)$.

\smallskip\noindent{\bf $\mathbold{\lang(G')\subseteq \lang(G,A,L)}$.}
It is easy to see from the productions of $G'$ that the only way to remove a starred non-terminal is to eventually replace a $A^*$ by $\epsilon$. By construction this is the last non-terminal in some prefix $w$ of a word in $G$ with start state $A$ and thus $w\in \lang(G,A,L)$.

\smallskip\noindent{\bf $\mathbold{\lang(G,A,L)\subseteq \lang(G')}$.}
Given a word $w$ in $\lang(G,A,L)$ by definition there is an implied production rule $A\rightarrow^* wAw'$ (in $G$) for some $w'$. 
Given a derivation tree $\dTree$ for the implied production rule $A\rightarrow^* wAw'$, 
it is easy to construct a derivation tree $\dTree'$ for $w$ in $\lang(G')$, indicating that $w$ is in $\lang(G')$. 
The two trees $\dTree$ and $\dTree'$ are identical except as follows:
For a node $v$ in $\dTree'$ let $\dTree(v)$ be the corresponding node in $\dTree$. 
 Let $\ell$ be the leaf in $\dTree'$ such that $\dTree(v)$ is the leaf with label $A$ in $\dTree$. 
 The production rule of $\ell$ is $A^*\rightarrow \epsilon$. Then, consider the path $\pi$ from $\ell$ to the root of $\dTree'$. 
 For each internal node $v$ in $\pi$ where $\dTree(v)$ has production rule $X\rightarrow BC$, we have the following:
\begin{itemize}
\item  {\bf $\mathbold{\pi}$ comes from the left child.}  If $\pi$ goes through the left child, the production rule of $v$ is $X^*\rightarrow B^*$ and the sub-tree under the right child of $\dTree(v)$ is cut out of $\dTree'$ (including that $v$ has no right child in this case).
\item {\bf $\mathbold{\pi}$ comes from the right child.} If $P$ goes through the right child the production rule of $v$ is $X^*\rightarrow BC^*$.
\end{itemize} 
Then, tree $\dTree'$ spells the word $w$ and is a derivation tree in $G'$. Thus $w\in \lang(G')$ and the lemma follows.
\end{proof}

\smallskip\noindent{\bf Realizable. }
Given a CFG $G$ in Chomsky normal form, we define the {\em realizable CFG } $\bar{G}$ (for clarity we do not define it in CNF) that 
\begin{itemize}
\item for each production of the form $P:A\rightarrow a$ in $G$ have the production $P:A\rightarrow \epsilon$,
\item for each production of the form $P:A\rightarrow BC$ in $G$ have the production $P:A\rightarrow a_L^A B C a_R^A$
\end{itemize}
 and no other productions (the language is then especially over the terminals $\{a_D^A\mid D\in \{L,R\}\wedge A\textrm{ is a non-terminal}\}$). 
A sequence is {\em realizable} if it is a sub-sequence of a word in $\lang(\bar{G})$, i.e., it results from deletion of letters from some word of $\lang(\bar{G})$.

\begin{lem}\label{lem:de_word}
Let $a_{D_1}^{A_1} \ldots a_{D_k}^{A_k}$ be a realizable sequence in $G$.
Then for every sequence of words $w_1 \in \lang(G,A_1,D_1), \ldots, w_k \in \lang(G,A_k,D_k)$
there exist words $s_1, \ldots, s_k, s_{k+1}$ such that
$s_1 w_1 \ldots s_k w_k s_{k+1}$ belongs to $\lang(G)$.
\label{l:meaning-of-realizability}
\end{lem} 
\begin{proof}
\newcommand{\reSeq}{\alpha}
Let $\reSeq = a_{D_1}^{A_1} \ldots a_{D_k}^{A_k}$.
We consider two cases: Either $\reSeq \in\lang(\bar{G})$ or not.

 \smallskip\noindent{\bf The case where $\mathbold{\reSeq \in \lang(\bar{G})}$. }
Consider a derivation tree $\dTree$ for $\reSeq$. We translate it into a derivation tree in $G$ for 
$s_1 w_1 \ldots s_k w_k s_{k+1}$, by replacing each production (which are in $\bar{G}$) of the nodes in $\dTree$ with (generalized) productions in $G$.
 
Each leaf node $v$ corresponds to a production $P:A\rightarrow \epsilon$. By definition there exists a production $P:A\rightarrow a$ in $G$ and we then simply replace $P$ in $\bar{G}$ with $P$ in $G$.

Each non-leaf node $v$, with children $b$ and $c$ respectively, corresponds to the use of a production $P:A\rightarrow a_L^A BC a_R^A$, where the $a_L^A$ is the $i$-th letter and $a_R^A$ the $j$-th of $w$ for some $i,j$.
By definition of $\lang(G,A,D)$ we have that there is a production $P':A\rightarrow^* w_i w_j' BC w_j w_i'$ in $G$ for some $w_i',w_j'$. In this case we replace $P$ in $\bar{G}$ with $P'$ in $G$.
(The word $s_{i+1}$ are concatenation of words $w_j'$ and letters derived  by productions $P:A\rightarrow a$ corresponding to $P:A\rightarrow \epsilon$. )

 \smallskip\noindent{\bf The case where $\mathbold{\reSeq \not\in \lang(\bar{G})}$. }
Find a word $\reSeq' \in\lang(\bar{G})$ such that $\reSeq$ is a sub-sequence of $\reSeq'$ (letting $p$ be the sequence of positions defining $\reSeq$ from $\reSeq'$) and do as above with $\reSeq'$ and the sequence of words $s'$ which is an extension of the sequence $s$ of length $|\reSeq'|$ by inserting $\epsilon$ at the remaining positions (i.e., the extension is such that $s$ is the sub-sequence of $s'$ defined by $p$).
\end{proof}

\smallskip\noindent{\bf Compact $\mathbold{G}$-decomposition. }
Given a CFG $G$ with a set of non-terminals of size $\nonTerm$ and a word $w \in \lang(G)$, we define 
a \emph{compact $G$-decomposition} of $w$ as  
$w=s_1 u_1 \ldots s_k u_k s_{k+1}$
 such that
\begin{enumerate}
\item for each $u_i$, there is an associated terminal $a_{D_i}^{A_i}$, such that the sequence $a_{D_1}^{A_1} \ldots a_{D_k}^{A_k}$ is realizable and $u_i\in \lang(G,A_i,D_i)$.
\item for all $\ell \in \N$, the word $w_\ell:=s_1 (u_1)^{\ell} s_2 \ldots s_k (u_k)^{\ell} s_{k+1}$ is in $\lang(G)$. 
\item $|w_0|=\sum_{i=1}^{k+1} |s_i| \leq 2^{\nonTerm}$ and $k \leq 2^{\nonTerm+1}-2$.
\end{enumerate}

\begin{lem}\label{lem:comp_deco}
For every CFG $G$ in CNF, every word $w \in \lang(G)$ admits a compact $G$-decomposition.
\end{lem}
\smallskip\noindent{\em Intuition.} The proof follows by repeated applications of the principle behind the pumping lemma, until the part which is not pumped is small.
\begin{proof}
Fix some $\ell$ and consider some word $w$ in $\lang(G)$ and some derivation tree $\dTree$ for $w$.
 We will greedily construct a compact $G$-representation, using that we do not give bounds on $|u_i|$. 

 \smallskip\noindent{\bf Greedy traversal and the first two properties. }
 The idea is to consider nodes of $\dTree$ in a depth first pre-order traversal (ensuring that when we consider some node we have already considered its ancestors). 
 When we consider some node $v$, we continue with the traversal, 
 unless there exists a descendant $u$ of $v$, such that $\dTree[v]=\dTree[u]$. 
 If there exists such a descendant, let $u'$ be the bottom-most descendant (pick an arbitrary one if there are more than one such bottom-most descendants) such that $A:=\dTree[v]=\dTree[u']$. 
 We say that $(v,u')$ forms a {\em pump pair} of $w$.
Consider subword $\alpha_v, \alpha_{u'}$ of $w$ derived by subtrees of $\dTree$ with roots at $v$ and $u'$ respectively.
 We can then write $\alpha_v$ as $s \alpha_{u'} s'$
 (and hence $A\rightarrow_G^* s A s'$), for some $s$ and $s'$ in the obvious way and $s$ and $s'$ will correspond to $u_i$ and $u_j$ respectively for some $i<j$ ($i$ and $j$ are defined by the traversal that we have already assigned $i-1$ $u$'s then we first visit $v$ and then assign $s$ as the $u_i$ and then we return to the parent of $v$, we have assigned $j-1$ $u$'s and assign $s'$ to be $u_j$). 

 Furthermore, $u_i$ is associated with $a_L^A$ and $u_j$ is associated with $a_R^A$. 
Observe that $A\rightarrow^* u_i A u_j$ implies that $u_i\in \lang(G,A,L)$ and $u_j\in \lang(G,A,R)$ and we therefore have ensured the first property of compact $G$-representation. 
 This also shows that we can replace $u_i$ with $(u_i)^\ell$ and $u_j$ with $(u_j)^\ell$ (because, clearly $A\rightarrow^* (u_i)^\ell A (u_j)^\ell$) and the new word is in $\lang(G)$. 
 Hence, $w_\ell$ is in $\lang(G)$, showing the second property of compact $G$-representation.
 This furthermore defines a derivation tree $\dTree_0$ for $w_0$ (which has $0$ occurrences of words $u_1, u_2, \ldots$), 
 which is the same as $\dTree$, except that for each pump pair $(v,u')$, the node $v$ is replaced with the sub-tree of $\dTree$ with root $u'$.
So as to  not split $u_i$ or $u_j$ up, we continue the traversal on $u'$, which, when it is finished, continues the traversal in the parent of $v$, having finished with $v$.  Notice that this ensures that each node is in at most one pump pair.

  \smallskip\noindent{\bf The third property. }
  Consider the word $w_0$ which has $0$ occurrences of words $u_1, u_2, \ldots$. Observe that in derivation tree $\dTree_0$ for $w_0$, there is at most one occurrence of each non-terminal in each path to the root, since we visited all nodes of $\dTree_0$ in our defining traversal and were greedy. 
  Hence, the height is at most $\nonTerm$ and thus, since the tree is binary, it has at most $2^{\nonTerm-1}$ many leaves, which is then a bound on $|w_0|=\sum_{i=1}^{k+1} |s_i|$.
Notice that each node of $\dTree_0$, being a subset of $\dTree$, is in at most $1$ pump pair of $w$. On the other  hand for each pump pair $(v,u')$ of $w$, we have that $u'$ is a node of $\dTree_0$ by construction. Hence, $w$ has at most $2^{\nonTerm}-1$ many pump pairs. Since each pump pair gives rise to at most $2$ word $u_i, u_{i'}$, we have $k\leq 2^{\nonTerm+1}-2$.
\end{proof}

\smallskip\noindent{\bf Sets closed under reachability.}
Fix an NFA.
We say that a set $Q'$ of states of the NFA is \emph{closed under reachability} if for all $q\in Q'$ and $a \in \Sigma$ we have $\delta(q,a)\subseteq Q'$. Clearly, the set of all states is closed under reachability.

\smallskip\noindent{\bf Reachability sets.}
Fix an NFA.
Given a state $q$ in the NFA and a word $w$, let $Q_q^w$ be the set of states reachable upon reading $w$, starting in $q$. The set of states $\R(w,q)$ is then the set of states reachable from $Q_q^w$ upon reading any word. For a set $Q'$ and word $w$, the set $\R(w,Q')$ is $\bigcup_{q\in Q'} \R(w,q)$.

Note the following: For all $Q'$ and $w$ the set $\R(w,Q')$ is closed under reachability.
If a set $Q'$ is closed under reachability then $\R(w,Q')\subseteq Q'$ for all $w$. 

We have the following {\bf property of reachability sets}:  Fix a word $u$, a number $\ell$, an NFA and a set of states $Q'$ of the NFA, where $Q'$ is closed under reachability. Let $u'$ be a word with $\ell$ occurrences of $u$ (e.g. $u^\ell$).
Consider any word $w$ with edit distance strictly less than $\ell$ from $u'$. Any run on $w$, starting in some state of $Q'$, reaches a state of $\R(u,Q')$. This is because $u$ must be a sub-word of $w$.

\begin{lem}
Let $G$ be a CFG in CNF with a set of non-terminals of size $\nonTerm$ and let $\aut_N$ be a safety NFA with a set of states $Q$.
The following conditions are equivalent:
\begin{enumerate}[label=(\roman*)]
\item the edit distance $\ed(\lang(G),\lang({\aut_N}))$ is infinite,
\item the edit distance $\ed(\lang(G),\lang({\aut_N}))$ exceeds $B:=(2^{\nonTerm+1}-2)\cdot n+2^{\nonTerm}$, and
\item there exists a word $w \in \lang(G)$, with compact $G$-decomposition
$w=(s_i u_i)_{i=1}^{k}s_{k+1}$, such that
$\R(u_k, \R(u_{k-1}, \R(u_{k-2}, \ldots \R(u_1, Q) \ldots ))) = \emptyset$.
\item there exist words $u_1, \ldots, u_k$ such that 
$\R(u_k, \R(u_{k-1}, \R(u_{k-2}, \ldots \R(u_1, Q) \ldots ))) = \emptyset$ and
for every $\ell >0$ 
there exist words $s_1, \ldots, s_{k+1}$ such that 
the word $w_{{\ell}} = (s_i u_i^{\ell})_{i=1}^{k}s_{k+1} $ belongs to $\lang(G)$.
\end{enumerate}
\label{l:FED-equivalences}
\end{lem}

\noindent We use condition~(iv) from Lemma~\ref{l:FED-equivalences} later in Section~\ref{sec:FEDUpperBoundDFA}.
Before we proceed with we argue by example that the nested applications of the $\R$ function in Lemma~\ref{l:FED-equivalences}  is necessary.

\smallskip\noindent{\em The necessity of the recursive applications of the $\R$ operator.}
Consider for instance the alternate requirement that at least one of $\R(u_i,Q)$ is empty, for some $i$. 
This alternate requirement would not capture that 
the pushdown language $\{a^n \# b^n\mid n\in \N\}$ has infinite edit distance to 
the regular language $a^* + b^*$ --- for any word in the pushdown language $w=a^n\#b^n$, for some fixed $n$, a compact $G$-representation of $w$ is $u_1=a^{n}$, $s_2=\#$ and $u_2=b^n$ (and the remaining words are empty). 
But clearly $\R(u_1,Q)$ and $\R(u_2,Q)$ are not empty since both strings are in the regular language. On the other hand $\R(u_2,\R(u_1,Q))$ is empty.

\begin{proof}
The implication \textbf{(i) $\Rightarrow$ (ii)} is trivial.

We show the implication \textbf{(ii) $\Rightarrow$ (iii)} as follows: Consider a word $w \in \lang(G)$ with 
$\ed(w,\lang({\aut_N})) > B$ and its 
compact $G$ representation $w = (s_i u_i)_{i=1}^{k}s_{k+1}$ (which exists due to Lemma~\ref{lem:comp_deco}).
We claim that $\R(u_k, \R(u_{k-1}, \R(u_{k-2}, \ldots \R(u_1, Q) \ldots ))) = \emptyset$.
The argument is by contradiction. Assume that $\R(u_k, \R(u_{k-1}, \R(u_{k-2}, \ldots \R(u_1, Q) \ldots )))\neq \emptyset$ and we will construct a run of $\aut_N$ spelling a word $w'$ in $\lang({\aut_N})$, which has edit distance at most $B$ to $w$.
The description of the run is iteratively in $i$; we start with $i=0$.
First, spell out a word $s_i'$, so that $\aut_N$ reaches some state $q_i$ such that there exists a run on $u_i$. The length of $s_i'$ is at most $n$. Afterwards follow the run on $u_i$ and go to the next iteration. This run spells the word $w':=(s_i' u_i)_{i=1}^{k}$. All the choices of $q_i$'s can be made since $\R(u_k, \R(u_{k-1}, \R(u_{k-2}, \ldots \R(u_1, Q) \ldots )))\neq \emptyset$.
Also, since $\aut_N$ is a safety automata, this run is accepting.
To edit $w'$ into $w$ change each $s_i'$ into $s_i$ and insert $s_{k+1}$ at the end. In the worst case, each $s_i$ is empty except for $i=k+1$ and in that case it requires $k\cdot n+|w_0|\leq B$  edits for deleting each $s_i'$ and inserting $s_{k+1}$ at the end (in any other case, we would be able to substitute some letters when we change some $s_i'$ into $s_i$ which would make the edit distance smaller). This is a contradiction. 

The implication \textbf{(iii) $\Rightarrow$ (iv)} is trivial.

For the implication \textbf{(iv) $\Rightarrow$ (i)} we will argue that for all $\ell$, the word $w_\ell\in \lang(G)$ requires at least $\ell$ edits.
Consider $w_\ell=(s_i u_i^{\ell})_{i=1}^ks_{k+1}$ for some $\ell$.
Any run on $s_1 u_1^\ell$ (a prefix of $w_\ell$) has entered $\R(u_1,Q)$ or made at least $\ell$ edits by the property of reachability sets. Similarly, for any $j$, any run on $(s_i u_i^{\ell})_{i=1}^j$ has either entered $\R(u_j, \R(u_{j-1}, \R(u_{j-2}, \ldots \R(u_1, Q) \ldots )))$ or there has been at least $\ell$ edits. Since  $\R(u_k, \R(u_{k-1}, \R(u_{k-2}, \ldots \R(u_1, Q) \ldots ))) = \emptyset$, no run can enter that set and thus there has been at least $\ell$ edits on $w_\ell$. The implication and thus the lemma follows.
\end{proof}

As a direct consequence of Lemma~\ref{l:FED-equivalences} we have the following.

\begin{thm}
(1)~For a PDA $\aut_P$ and an NFA $\aut_N$ we have $ed(\lang(\aut_P),\lang(\aut_N))$ is either exponentially bounded in 
$|\aut_P|$ or it is infinite.
(2)~For $\class_1 \in \{\DPDA,\PDA\}$ we have $\FED(\class_1, \NFA)$ is in $\EXPTIME$
\end{thm}
\begin{proof}
(1)~The equivalence of (i) and (ii) gives a bound on the maximum finite edit distance. 

(2)~The argument follows from Lemma~\ref{l:TEDinexp} and (1), i.e., we can check with
Lemma~\ref{l:TEDinexp} $\TED$ for $k$ exceeding the bound from (1). 
\end{proof}

\subsection{Upper bound for DFA}
\label{sec:FEDUpperBoundDFA}

We show that the problem $\FED(\class_1, \DFA)$ is $\coNP$-complete for $\class_1 \in \{\DPDA,\PDA\}$.

\smallskip\noindent{\bf $\mathbold{\coNP}$-hardness and attempting to apply known techniques for the upper bound.}
The lower bound follows directly from the fact that $\FED(\DFA, \DFA)$ is $\coNP$-hard~\cite{boundedRiveros}. We thus focus on the upper bound. Note that the upper bound was simple for $\FED(\DFA, \DFA)$, since the edit distance for such is either polynomial or infinite and there is a polynomial length witness in case it is infinite. Hence, one just guess the polynomial sized witness~$w$ and runs a polynomial time algorithm for $\ed(w,\DFA)$ and the result follows. Doing the similar thing for $\FED(\PDA, \DFA)$ would give a $\NEXPTIME$ upper-bound, since the word we need to guess might be of exponential length (thus the above $\EXPTIME$ upper bound for $\FED(\PDA, \NFA)$ is better). To give our algorithm, we will first define extended reachability sets and give a key proposition.

\smallskip\noindent{\bf Closed under concatenation and extended reachability sets.}
A language $L$ is said to be {\em closed under concatenation} if for all $w_1,w_2\in L$ we have $w_1w_2\in L$. Note that $\lang(G,A,D)$, for any non-terminal $A$ and direction $D$, is always closed under concatenation.

We extend reachability sets as follows:
Let $L$ be a context-free language closed under concatenation, let $\aut_D$ be a DFA and let $Q'$ be a subset of the states of $\aut_D$.
We define $\R(L, Q')$ as the intersection $\bigcap_{w \in L} \R(w, Q')$. 
Observe that for every $L$ there exists a finite subset $W \subseteq L$ such that $\bigcap_{w \in W} \R(w, Q') = \R(L, Q')$.

\begin{rem}\label{rem:R_set}
 If $Q'$ is closed under reachability, then for any set
 $W=\{w_1,w_2,\dots,w_k\}\subseteq L$ of words such that~$\bigcap_{w \in W} \R(w, Q')=\R(L,Q')$, we have that $w'=w_1w_2\dots w_k\in L$ and  $\R(w',Q')=\R(L,Q')$. The latter comes from the fact that for any word $w''$ and set $Q''$ closed under reachability, we have that $\R(s_1w''s_2,Q'')\subseteq \R(w'',Q'')$ for all $s_1$ and $s_2$. 
 
Also, observe that we have the following facts about $\R$, from the definition of $\R$:
\begin{enumerate}
\item For any $Q''\subseteq Q'$ and word $w$ we have that $\R(w,Q'')\subseteq \R(w,Q')$.
\item For any language $L$, any $Q'$ and word $w\in L$, we have that $\R(L,Q')\subseteq \R(w,Q')$.
\end{enumerate}
  \label{rem:exists_word_equal_to_L}
\end{rem}

The following proposition is a key to our $\coNP$-algorithm.

\begin{proposition}\label{pro:language_to_words}
For any $k$, any sequence of languages $L_1,\dots,L_k$ and any word $w_i\in L_i$ for each $i$, we have 
\[
\R(L_k,  \ldots, \R(L_1, Q) \ldots )) \subseteq \R(w_k, \ldots, \R(w_1, Q) \ldots ))) \enspace .
\]
Also, if each $L_i$ is closed under concatenation, then there exist words $w_i'\in L_i$ for each $i$, such that
\[
\R(L_k,  \ldots, \R(L_1, Q) \ldots )) = \R(w_k', \ldots, \R(w_1', Q) \ldots )))
\]
\end{proposition}
\begin{proof}

The proposition 
follows from Remark~\ref{rem:exists_word_equal_to_L} and simple induction.
\end{proof}

\smallskip\noindent{\bf $\mathbold{\coNP}$-upper bound algorithm.} 
Our $\coNP$-algorithm \algoFEDPDADFA\ deciding whether the edit distance is finite works as follows:
\begin{enumerate}
\item Guess a sequence $s=a_{D_1}^{A_1}a_{D_2}^{A_2}\dots a_{D_k}^{A_k}$, for some $k$.
\item return ``no'' if $s$ is such that (1)~$s$ is realizable; and (2) \[\R(\lang(G,A_k,D_k), \R(\lang(G,A_{k-1},D_{k-1}),  \ldots \R(\lang(G,A_1,D_1), Q) \ldots ))) = \emptyset\enspace .\]
\item otherwise return yes.
\end{enumerate}

\smallskip\noindent{\bf Requirements for \algoFEDPDADFA\ to be in $\mathbold{\coNP}$.}
For \algoFEDPDADFA\ to be in $\coNP$, we need to give the following:
\begin{enumerate}
\item A polynomial bound on $k$ (so that $s$ is a polynomial sized witness). The bound will be given in Lemma~\ref{l:bound_k}.
\item A polynomial time algorithm to decide whether a sequence $s=a_{D_1}^{A_1}a_{D_2}^{A_2}\dots a_{D_k}^{A_k}$ is realizable. The algorithm will be given in Lemma~\ref{l:realizability-polynomial}.
\item A polynomial time algorithm for computing  $\R(\lang(G,A,D), Q')$ for any CFG $G$, any non-terminal $A$, any direction $D$ and any set  $Q'$ closed under reachability.
This will allow us to decide, given a realizable sequence $s=a_{D_1}^{A_1}a_{D_2}^{A_2}\dots a_{D_k}^{A_k}$, whether 
\[\R(\lang(G,A_k,D_k), \R(\lang(G,A_{k-1},D_{k-1}),  \ldots \R(\lang(G,A_1,D_1), Q) \ldots ))) = \emptyset\enspace ,\] 
by evaluating the expression on the left-hand side inside-out. The algorithm for computing  $\R(\lang(G,A,D), Q')$ will be given in Corollary~\ref{cor:compute-reachable-states}.
\end{enumerate}

We will first argue that the algorithm is correct.

\begin{lem}
The algorithm \algoFEDPDADFA\ is correct\label{l:correct}.
\end{lem}

\begin{proof}
To argue that the algorithm is correct, we just need to argue that a sequence with properties (1) and (2) exists if and only if the edit distance is infinite.

\smallskip\noindent{\bf Such a sequence implies infinite edit distance.}
According to Proposition~\ref{pro:language_to_words}, such a sequence indicates that there are words $w_i\in \lang(G,A_i,D_i)$ for each $i$, such that 
\[\R(w_k, \ldots, \R(w_1, Q) \ldots ))) = \emptyset\enspace .\]
For all $i$, since $\lang(G,A_i,D_i)$ is closed under concatenation, we also have $w_i^\ell\in \lang(G,A_i,D_i)$ for all $\ell>0$.
Thus, by Lemma~\ref{l:meaning-of-realizability}, there exist words $s_1, \ldots, s_k, s_{k+1}$ such that
$s_1 w_1^\ell \ldots s_k w_k^\ell s_{k+1}$ belongs to $\lang(G)$.
Hence, item (iv) of Lemma~\ref{l:FED-equivalences} is satisfied and we get that the edit distance is infinite. 

\smallskip\noindent{\bf Infinite edit distance implies the existence of such a sequence.}
When the edit distance is infinite, according to Lemma~\ref{l:FED-equivalences}(iii) there exists a word $w \in \lang(G)$, with compact $G$-decomposition
$w=(s_i u_i)_{i=1}^{k}s_{k+1}$, such that
$\R(u_k, \R(u_{k-1}, \R(u_{k-2}, \ldots \R(u_1, Q) \ldots ))) = \emptyset$. 
By definition of compact $G$-decomposition, every $u_i$ from the decomposition is associated with a terminal $a_{D_i}^{A_i}$, such that the sequence $a_{D_1}^{A_1} \ldots a_{D_k}^{A_k}$ is realizable (satisfying property (1)) and $u_i\in \lang(G,A_i,D_i)$ for each $i$. By Proposition~\ref{pro:language_to_words} we then have that \[\R(\lang(G,A_k,D_k), \R(\lang(G,A_{k-1},D_{k-1}),  \ldots \R(\lang(G,A_1,D_1), Q) \ldots ))) = \emptyset\enspace ,\] (satisfying property (1)). Thus such a sequence always exists and the lemma follows.
\end{proof}

Next, we will give the bounds and algorithms to show that \algoFEDPDADFA\ is in $\coNP$. First the bound on $k$. 
\begin{lem}\label{l:bound_k}
Let $G$ be a CFG and let $\aut_D$ be a safety DFA with a set of states $Q$.
The following conditions are equivalent:
\begin{enumerate}[label=(\roman*)]
\item the edit distance $\ed(\lang(G),\lang({\aut_D}))$ is infinite.
\item there exists a realizable sequence $(a_{D_i}^{A_i})_{i=1}^m$ with $m \leq |Q|$
such that \[\R(\lang(G,A_m, D_m) , \R(\lang(G,A_{m-1}, D_{m-1}), \ldots, \R(\lang(G,A_1, D_1), Q) \ldots )) = \emptyset\enspace .\]
\end{enumerate}
\end{lem}
\begin{proof}
\noindent{\textbf{(i) implies (ii).}}
Assume that $\ed(\lang(G),\lang({\aut_D}))$ is infinite. By Lemma~\ref{l:FED-equivalences},
there exists a word $w \in \lang(G)$, with compact $G$-decomposition
$w=(s_i u_i)_{i=1}^{k}s_{k+1}$, such that
$\R(u_k, \R(u_{k-1}, \ldots, \R(u_1, Q) \ldots )) = \emptyset$. Observe that $k$ can be exponential. We claim that we can 
pick from $u_1, \ldots, u_k$ a sub-sequence 
of polynomial length in $|Q|$ for which the reachable set of states is empty as well.
Indeed, the sequence $s=\R(u_1, Q), \R(u_2, \R(u_1, Q)), \ldots$
is weakly decreasing with respect to the set inclusion (i.e. if a state is not in $s[i]$, then, it cannot be in $s[j]$ for $j\geq i$, because $\R$ is closed under reachability). 
We select from $1,\ldots,k$ indices $i$ on which the sequence  $\R(u_1, Q), \R(u_2, \R(u_1, Q)), \ldots$ 
strictly decreases and denote the resulting sub-sequence by $\alpha$. 
Then, \[\R(u_{\alpha(m)}, \R(u_{\alpha(m-1)},\ldots, \R(u_{\alpha(1)}, Q) \ldots ))) = \emptyset\enspace .\]
There are at most $|Q|$ such indices, therefore $|\alpha| = m \leq |Q|$.
Using Proposition~\ref{pro:language_to_words}, since $u_{\alpha(i)}\in \lang(G,A_{\alpha(i)},D_{\alpha(i)})$ by compact $G$-decomposition,  we get that
\begin{align*}
\R(\lang(G,A_{\alpha(m)}, D_{\alpha(m)}) ,  \ldots, \R(\lang(G,A_{\alpha(1)}, D_{\alpha(1)}), Q) \ldots )) \subseteq \R(u_{\alpha(m)}, \ldots, \R(u_{\alpha(1)}, Q) \ldots )))
\end{align*}
and hence is empty.

\noindent{\textbf{(ii) implies (i).}} 
Assume that condition~(ii) holds. Then, the algorithm \algoFEDPDADFA\ returns YES, and its correctness (Lemma~\ref{l:correct}) implies (i).
\end{proof}

Next we will describe the algorithm deciding whether a sequence is realizable.

\begin{lem}
Let $G$ be a CFG.  
We can decide in polynomial time whether
a given sequence $s=a_{D_1}^{A_1}\ldots a_{D_k}^{A_k}$ is realizable.
\label{l:realizability-polynomial}
\end{lem}
\begin{proof}
Consider grammar $\bar{G}$ associated with $G$. 
We convert $\bar{G}$ to CNF and add productions $A \rightarrow \epsilon$ for every non-terminal. 
Let the resulting CFG be $G'$.
Observe that $G'$ derives a word of terminals and non-terminals $u$ if and only if 
$\bar{G}$ derives a word $u'$ such that $u$ is a subsequence of $u'$. 
Thus, $(A_1,D_1), \ldots, (A_k, D_k)$ is realizable if and only if 
$A_1^{D_1} \ldots A_k^{D_k}$ is derivable by $G'$. 
Since $G'$ has polynomial size in $G$, we can check whether a word is derivable in $G'$ in polynomial time.
\end{proof}

Finally, we present the algorithm that computes $\R(\lang(G,A,D),Q')$. 
The result will follow as a corollary of the following lemma.

\begin{lem}
Given a CFG $G$, such that $\lang(G)$ is closed under concatenation, a DFA $\aut_D$ with a set of states $!$ and a set of states $Q' \subseteq Q$ closed under reachability, 
the set $\R(\lang(G), Q')$ is computable in polynomial time.
\label{l:compute-reachable-states}
\end{lem}
\begin{proof}
Given a set of states $S \subseteq Q$, we define $\reach(S)$ as the set of states reachable from $Q'$ in $\aut_D$.
We can divide $Q'$ into strongly connected components (SCCs). 
We say that an SCC $C$ is \emph{recurrent} if $\aut_D$ can stay in $C$ upon reading any word from $\lang(G)$.

We claim that $\R(\lang(G), Q')$ is the set of states $Q^*$ reachable from all recurrent SCCs in $Q'$. Clearly, $Q^*$ is closed under reachability.
\begin{itemize}
\item We will first argue that $Q^*\subseteq \R(\lang(G), Q')$.
First, for every recurrent SCC $C$ and every word $w \in \lang(G)$, there is a state $s' \in C$ such that 
$\R(w, s') \in C$. Therefore, $C \subseteq \R(w, s')$. By Remark~\ref{rem:exists_word_equal_to_L}, it follows that $C \subseteq \R(\lang(G), Q')$ and $Q^*\subseteq \R(\lang(G), Q')$. 

\item We will next argue that $\R(\lang(G), Q')\subseteq Q^*$
Observe that for every state $s$ in a non-recurrent SCC $C$ there exists a word $w_s \in \lang(G)$ that forces $\aut_D$ to leave $C$, i.e.,
$\R(w_s, s) \cap C =\emptyset$. Thus, $|\R(w_s, C) \cap C| < |C|$. It follows that we can remove states from $\R(w_s, C) \cap C$
one by one by concatenating words $w_s$ to obtain a word $w_C$ such that $\R(w_C, C) \cap C =\emptyset$.
Since $\lang(G)$ is closed under concatenation, the word $w_C$ belongs to $\lang(G)$.

Let $C_1,C_2,\dots, C_\ell$ be the SCCs in $Q'$ not in $Q^*$ (and thus non-recurrent) ordered topologically. Let the word $w_T$ be the word $w_T=w_{C_1}w_{C_2}\dots w_{C_\ell}$.  Observe that $w_T\in \lang(G)$. 
We have $\R(\lang(G),Q')\subseteq \R(w_T,Q')$ by Remark~\ref{rem:exists_word_equal_to_L} and we argue that $\R(w_T,Q')\subseteq Q^*$. 

Any run starting in $Q^*$ will end in $Q^*$, since $Q^*$ is closed under reachability. 
Observe that $\R(w_{C_1}, C_1 \cup \ldots \cup C_{\ell})$ does not contain $C_1$ as 
$\R(w_{C_1}, C_1) \cap C_1 = \emptyset$ and due to topological order $C_1$ is not reachable from $C_2, \ldots, C_{\ell}$.
Thus, by induction reasoning we have $\R(w_{C_1}, \R(w_{C_2}, \ldots, \R(w_{C_\ell}, C_\ell) \ldots )$ does not 
contain $C_1, \ldots, C_\ell$. Observe that 
$\R(w_T, C_1 \cup \ldots \cup C_\ell) \subseteq \R(w_{C_1}, \R(w_{C_2}, \ldots, \R(w_{C_\ell}, C_\ell) \ldots )$, 
and hence $\R(w_T, C_1 \cup \ldots \cup C_\ell) \subseteq Q^*$.

\end{itemize}

\noindent Given a SCC $C$ and a state $s\in C$, let the automaton $\aut_D^{C,s}$ be $\aut_D$ restricted to $C$ and with start state $s$.
Observe that a SCC $C$ is recurrent if and only if there is a state $s \in C$ such that $\lang(G)\subseteq \lang(\aut_D^{C,s})$. We can then easily test if a SCC is recurrent by trying each possibility for $s\in C$ and testing if $\lang(G)\subseteq \lang(\aut_D^{C,s})$. This can be done in polynomial time since language inclusion of a CFG in a DFA can be tested in polynomial time.

Thus our algorithm is as follows: Compute the set $\{C_1,\dots,C_\ell\}$ of SCCs in $Q'$. For each $i$ test if $C_i$ is recurrent and let $\{C_1',\dots,C'_{\ell'}\}$ be the recurrent SCCs in $Q'$. Return $\bigcup_{i=1}^{\ell'}\reach(C_i')$.
\end{proof}

We next get the wanted corollary.
\begin{cor}\label{cor:compute-reachable-states}
Given a CFG $G$, a non-terminal $A$, a direction $D$, a DFA $\aut_D$ and a set of states $Q'$ of $\aut_D$ closed under reachability, 
the set $\R(\lang(G,A,D), Q')$ is computable in polynomial time.
\end{cor}
\begin{proof}
The proof follows from Lemma~\ref{lem:L(G,A,D)} and Lemma~\ref{l:compute-reachable-states}, using that $\lang(G,A,D)$ is closed under concatenation.
\end{proof}

\begin{lem}
Given a context-free grammar $G$ and a safety DFA $\aut_D$ the algorithm \algoFEDPDADFA\ can be implemented in $\coNP$ and correctly decides whether $\ed(\lang(\aut_P), \lang(\aut_D))$ is finite.
Moreover, if $\aut_D$ is of constant size then \algoFEDPDADFA\ does not need non-determinism (and thus uses polynomial time only).
\label{l:FEDforSafeDFAiscoNP}
\end{lem}
\begin{proof}
The correctness comes from Lemma~\ref{l:correct}. The complexity comes from Lemma~\ref{l:bound_k}, Lemma~\ref{l:realizability-polynomial} and Corollary~\ref{cor:compute-reachable-states}. Note that, in case the DFA is of constant size, then $k$ is bounded by a constant, according to Lemma~\ref{l:bound_k}, and thus there are only a polynomial number of candidates for $s$ and hence all can be checked using polynomial time in total.
\end{proof}

\begin{thm}
For $\class_1 \in \{\DPDA,\PDA\}$ we have $\FED(\class_1, \DFA)$ is $\coNP$-complete.
\label{th:FEDonDFAcoNP}
\end{thm}
\begin{proof}
First, we discuss containment of $\FED(\PDA, \DFA)$ in $\coNP$.
Consider a PDA $\aut_P$ and a DFA $\aut_D$. We can transform $\aut_P$
to a context-free grammar $G$ with $\lang(\aut_P) = \lang(G)$ in polynomial time. 
Also, we can transform $\aut_D$ to a safety DFA $\aut_D'$ recognizing
the language $\prefix{\lang({\aut_D})}$.  
Due to Lemma~\ref{lem:prefix_closure}, we have 
$\ed(\lang(G),\aut_D)$ is finite if and only if $\ed(\lang(G),\aut_D')$ is finite. 
By Lemma~\ref{l:FEDforSafeDFAiscoNP} we can decide whether $\ed(\lang(G),\aut_D')$ is finite
in $\coNP$. Hence, $\FED(\PDA, \DFA)$ and $\FED(\DPDA, \DFA)$ are in $\coNP$.

Is has been shown in~\cite{boundedRiveros} that $\FED(\DFA, \DFA)$ is $\coNP$-hard, therefore
$\FED(\PDA, \DFA)$ and $\FED(\DPDA, \DFA)$ are $\coNP$-hard
\end{proof}

\subsection{Lower bound}
\label{sec:FEDLowerBound}

We have shown the exponential upper bound on the edit distance if it is finite.
As mentioned in the introduction, it is easy to define a family of context free grammars only accepting an exponential length word, using repeated doubling and thus the edit distance can be exponential between DPDAs and DFAs.
We can also show that the inclusion problem reduces to the finite edit distance problem $\FED(\DPDA,\NFA)$ and get the following lemma.

\begin{lem}
$\FED(\DPDA,\NFA)$ is $\EXPTIME$-hard.
\label{th:FEDexpHard}
\end{lem}
\begin{proof}
We show that the inclusion problem of $\DPDA$ in $\NFA$, which is $\EXPTIME$-hard by Lemma~\ref{l:ExpTimeHardness} reduces to 
$\FED(\DPDA,\NFA)$. 
Consider  a DPDA $\aut_P$ and an NFA $\aut_N$.
We define $\widehat{\lang} = \{ \# w_1 \# \ldots \# w_k \# : k \in \N, w_1, \ldots, w_k \in \lang\}$.
Observe that either  $\widehat{\lang_1} \subseteq \widehat{\lang_2}$ or $\ed(\widehat{\lang_1},\widehat{\lang_2}) = \infty$.
Therefore, $\ed(\widehat{\lang_1},\widehat{\lang_2}) < \infty$ if and only if 
$\lang_1 \subseteq \lang_2$. 
In particular, $\lang(\aut_P) \subseteq \lang(\aut_N)$ if and only if
$\ed(\widehat{\lang(\aut_P)},\widehat{\lang(\aut_N)}) < \infty$.
Observe that in polynomial time we can transform $\aut_P$ (resp., $\aut_N$) to 
a DPDA $\widehat{\aut_P}$ (resp., an NFA $\widehat{\aut_N}$) recognizing
$\widehat{\lang(\aut_P)}$ (resp., $\widehat{\lang(\aut_P)}$). 
It suffices to add transitions 
from all final states to all initial states with the letter $\#$, i.e.,
$\{ (q,\#,s) : q \in F, s \in S \}$ for NFA (resp., $\{ (q,\#, \bot,s) : q \in F, s \in S \}$ for DPDA).
For DPDA the additional transitions are possible only with empty stack.
\end{proof}

\section{Edit distance to PDA}
\makeatletter{}\label{s:fromPDA}

Observe that the threshold distance problem from $\DFA$ to $\PDA$ with the fixed threshold $0$
and a fixed DFA recognizing $\Sigma^*$ coincides with the universality problem for $\PDA$. 
Hence, the universality problem for $\PDA$, which is undecidable, reduces to $\TED(\DFA, \PDA)$.
The universality problem for $\PDA$ reduces to $\FED(\DFA, \PDA)$ as well by 
the same argument as in Lemma~\ref{th:FEDexpHard}. 
Finally, we can reduce the inclusion problem from $\DPDA$ in $\DPDA$, which is undecidable, to
$\TED(\DPDA, \DPDA)$ (resp., $\FED(\DPDA, \DPDA)$). Again, we can use 
the same construction as in Lemma~\ref{th:FEDexpHard}. 
In conclusion, we have the following proposition.

\begin{proposition}
(1)~For every class $\class \in \{ \DFA, \NFA, \DPDA, \PDA \}$, the problems
$\TED(\class, \PDA)$ and $\FED(\class, \PDA)$ are undecidable.
(2)~For every class $\class \in \{ \DPDA, \PDA \}$, the problem
$\FED(\class, \DPDA)$ is undecidable.
\label{p:undecidable}
\end{proposition}

The results in (1) of Proposition~\ref{p:undecidable} are obtained by reduction from
the universality problem for $\PDA$. However, the universality problem for $\DPDA$ is 
decidable. Still we show that $\TED(\DFA,\DPDA)$ is undecidable.
The overall argument is similar to the one in Section~\ref{sec:lowerBoundTED}.
First, we define nearly-deterministic PDA, a pushdown counterpart of nearly-deterministic NFA.

\begin{defi}
A PDA $\aut = (\Sigma, \Gamma, Q, S, \delta, F)$ is \emph{nearly-deterministic} if 
$|S| = 1$ and $\delta = \delta_1 \cup \delta_2$, where $\delta_1$ is a function and for
every accepting run, the automaton takes a transition from $\delta_2$ exactly once.
\end{defi}

By carefully reviewing the standard reduction of the halting problem for Turing machines to 
the universality problem for pushdown automata~\cite{HU79}, we observe that
the PDA that appear as the product of the reduction are 
nearly-deterministic.

\begin{lem}
The problem, given a nearly-deterministic PDA $\aut_P$, decide
whether $\lang({\aut_P}) = \Sigma^*$, is undecidable.
\label{l:PDAuniversality}
\end{lem}

Using the same construction as in Lemma~\ref{th:TEDexpHard} we show
a reduction of the universality problem for nearly-deterministic PDA to
$\TED(\DFA, \DPDA)$.

\begin{proposition}
For every class $\class \in \{ \DFA, \NFA, \DPDA, \PDA \}$,
the problem $\TED(\class, \DPDA)$ is undecidable.
\label{th:fromDPDAUndecidable}
\end{proposition}
\begin{proof} 
We show that $\TED(\DFA, \DPDA)$ (resp., $\FED(\DFA, \PDA)$) is undecidable as it implies undecidability of the rest of the problems.
The same construction 
as in the proof of Lemma~\ref{th:TEDexpHard} shows a reduction of 
the universality problem for nearly-deterministic PDA, which is undecidable by Lemma~\ref{l:PDAuniversality}, to $\TED(\DFA, \DPDA)$.
\end{proof}

We presented the complete decidability picture for the problems $\TED(\class_1, \class_2)$, for
 $\class_1 \in \{ \DFA, \NFA, \DPDA, \PDA \}$ and $\class_2 \in \{\DPDA, \PDA\}$. 
To complete the characterization of the problems $\FED(\class_1, \class_2)$,
with respect to their decidability, we still need to settle the decidability (and complexity)
status of $\FED(\DFA,\DPDA)$. We leave it as an open problem, but conjecture that it is $\coNP$-complete. 

\begin{conj}
$\FED(\DFA,\DPDA)$ is $\coNP$-complete.
\label{conj:FEDisUndec}
\end{conj}

\section{Conclusions}
In this work we consider the edit distance problem for PDA and its subclasses
and present a complete decidability and complexity picture for the $\TED$ problem.
We leave some open conjectures about the parametrized complexity of the $\TED$
problem, and the complexity of $\FED$ problem when the target is a DPDA.
Moreover, one can study the edit distance for other classes of languages between 
regular languages and context-free languages such as visibly pushdown automata.

While in this work we count the number of edit operations, a different notion is
to measure the average number of edit operations. 
The average-based measure is undecidable in many cases even for finite automata,
and in cases when it is decidable reduces to mean-payoff games on graphs~\cite{limavgRiveros}.
Since mean-payoff games on pushdown graphs are undecidable~\cite{CV12}, most of the problems 
related to the edit distance question for average measure for DPDA and PDA are likely
to be undecidable.

\noindent\textbf{Acknowledgements}. We wanted to thank the anonymous reviewers for their thorough and helpful
reviews, which help us to improve this paper.

\bibliographystyle{plain}
\bibliography{papers}

\end{document}